\documentclass[11pt,reqno]{amsart}
\usepackage{amsfonts,amssymb,amsmath,amsopn,amsthm,graphicx}
\usepackage{amsxtra, mathrsfs}
\usepackage{colordvi}
\usepackage[usenames,dvipsnames]{color}
\usepackage{amsfonts,amssymb,amsbsy,amsmath,amsthm,dsfont}
\usepackage{mathtools}
\usepackage[colorlinks=true]{hyperref}
 
\usepackage{enumitem}

\newcommand\version\today

\usepackage[parfill]{parskip}    

\numberwithin{equation}{section}

\newtheorem{theorem}{Theorem}[section]
\newtheorem{corollary}[theorem]{Corollary}
\newtheorem{lemma}[theorem]{Lemma}
\newtheorem{proposition}[theorem]{Proposition}

\theoremstyle{definition}
\newtheorem{definition}[theorem]{Definition}
\newtheorem{remark}[theorem]{Remark}
 
\newtheorem{assumption}[theorem]{Assumption}

\newcounter{theoremi}[theorem]

\newcounter{assumptions}

\newcounter{smalllist}

\newcounter{listi}

\newcounter{smallenum}

\newcommand{\Ra}{\big\rangle}
\newcommand{\La}{\big\langle}
\newcommand{\la}{\langle}
\newcommand{\ra}{\rangle}
\newcommand{\veps}{\varepsilon}

\newcommand{\id}{\mathds{1}}

\newcommand{\eps}{\varepsilon}
\newcommand{\N}{\mathbb{N}}
\newcommand{\R}{\mathbb{R}}

\newcommand{\x}{\la x \ra }
\newcommand{\rt}{{\rm curl}}
\newcommand{\C}{\mathbb{C}}

\newcommand{\D}{\mathscr{D}}
\newcommand{\h}{\mathcal{H}}

\newcommand{\U}{\mathcal{U}}

\newcommand{\pd}{\partial}

\newcommand{\LL}{\mathcal{\big\langle}}
\newcommand{\RR}{\mathcal{\big\rangle}}

\DeclareMathOperator{\im}{Im}

\DeclareMathOperator{\re}{Re}

\DeclareMathOperator{\supp}{supp}

\DeclareMathOperator*{\esssup}{ess\,sup}

\newcommand{\wti}{\widetilde  }
\newcommand{\curl}{{\rm curl}} 
\newcommand{\ol}[1]{\overline{#1}}

\title[Absence of embedded eigenvalues of Pauli and Dirac  operators]{Absence of embedded eigenvalues of Pauli and Dirac  operators}

\author {Dirk Hundertmark}

\address{Dirk Hundertmark, Department of Mathematics, Institute for Analysis, Karlsruhe Institute of Technology, 76128 Karlsruhe, Germany, and
Department of Mathematics
University of Illinois at Urbana-Champaign
1409 W. Green Street Urbana, Illinois 61801-2975 }
\email {dirk.hundertmark@kit.edu}

\author {Hynek Kova\v{r}\'{\i}k}

\address {Hynek Kova\v{r}\'{\i}k, DICATAM, Sezione di Matematica, Universit\`a degli studi di Brescia,Via Branze 38 - 25123, Brescia, Italy}

\email {hynek.kovarik@unibs.it}

\thanks{\noindent\copyright 2023 by the authors. Faithful reproduction of this article, in its entirety, by any means is permitted for non-commercial purposes. }

\begin{document}

\begin{abstract}
We consider eigenvalues of the Pauli operator in $\R^3$ embedded in the continuous spectrum. In our main result we prove the absence of such eigenvalues above a
threshold which depends on the asymptotic behavior of the magnetic and electric field at infinity.  We show moreover that the decay conditions on the magnetic and electric field are sharp. Analogous results are obtained for purely magnetic Dirac operators.
\end{abstract}

\maketitle



\smallskip


\section{\bf Introduction and outline of the paper } \label{sec-intro}
In this paper we study the point spectrum of the Pauli operator in $L^2(\R^3,\C^2)$  formally given by 
\begin{equation} \label{pauli-formal}
H_{A,V} = \big (\sigma\cdot (P-A)\big)^2\,  + V.
\end{equation} 
Here $P=-i\nabla $ denotes the momentum operator, $A\in L^2_{\text{loc}}(\R^3,\R^3)$ is a magnetic vector potential, $\sigma=(\sigma_1,\sigma_2,\sigma_3)$ is the set of  Pauli matrices, see equation \eqref{sigma-j} below, and $V$ is a potential function which associates to each $x\in\R^3$ a two by two hermitian matrix $V(x)$. We refer to equations \eqref{q0} and \eqref{q} for a more precise definition of $H_{A,V} $. The free Pauli operator $H_{A,0}$ represents a quantum Hamiltonian of a particle with spin $\frac 12$ interacting with a magnetic field $B= \curl\, A$, see e.g.~\cite{thaller} for further reading and references. 

Our aim is to find sharp conditions on $B$ and $V$ under which the operator $H_{A,V}$ has no eigenvalues above certain critical energy.  

The absence of discrete eigenvalues of $H_{A,V}$, also in dimensions higher than three,  can be deduced from the results of \cite{cfk}, where the authors show, via the method of multipliers, that if $B$ and $V$ satisfy certain smallness assumptions, then $H_{A,V}$ has no eigenvalues at all.

The absence of eigenvalues at the threshold of the essential spectrum, typically zero, is also well understood, at least in the case $V=0$. A sharp criterion for zero to be an eigenvalue of $H_{A,0}$  was recently established in \cite{fl,fl2}, see also \cite{be,bvb}. In particular, it is proved in \cite{fl2} that $H_{A,0}$ can have a zero energy eigenfunction only if $\|A\|_{L^3(\R^3)}$ exceeds certain explicit value. 
Examples of magnetic fields which produce zero energy eigenfunctions of $H_{A,0}$, and which show that the criterion of \cite{fl2} is sharp,  can be found in \cite{amn,elt,ly,rs}.  We will give more comments on this question in Remark \ref{rem-zero}.

What is not well understood so far is the question of absence of eigenvalues embedded in the essential spectrum, which is of fundamental importance e.g.~for the validity of a limiting absorption principle, for the scattering theory, as well as for dispersive estimates.  One could of course apply the result of \cite{cfk}, since the conditions stated there guarantee not only the absence of discrete eigenvalues, but also the absence of {\it all eigenvalues}, \cite[Thm.~3.5]{cfk}. However, this automatically implies
that such conditions are way too strong if one is interested only in embedded eigenvalues, since creating discrete eigenvalues is usually much ``easier" than creating  eigenvalues embedded in the essential spectrum. Indeed, consider the generic case in which $\sigma_{\rm es} (H_{A,0})= [0,\infty)$. Then any negative and sufficiently strong potential $V$ will create negative eigenvalues, but it should typically not create positive eigenvalues, at least when $B$ and $V$ decay fast enough at infinity. Hence in order to exclude all eigenvalues, one
has to impose global smallness assumptions on $B$ and $V$, see \cite[Thm.~3.5]{cfk}. On the other hand, embedded eigenvalues belong to the essential spectrum and therefore their absence or existence should depend only on the behavior of $B$ and $V$ at infinity.

In this paper we prove that the operator $H_{A,V}$  cannot have eigenvalues above an energy level $\Lambda \geq 0$ allowing, at the same time, $H_{A,V}$ to have discrete and/or threshold eigenvalues, see  Theorem \ref{thm-abs}. We provide an explicit expression for $\Lambda$ which shows, in agreement with the above heuristics, that $\Lambda$ depends only on the behavior of $B$ and $V$ at infinity. In particular, no global bounds on $B$ and $V$ are needed.

Let us describe the main result of this paper more in detail. In Theorem \ref{thm-abs} it is proved, under rather mild regularity and decay conditions on $B$ and $V$, that  
$H_{A,V}$ has no eigenvalues larger than 
\begin{equation} \label{edge}
 \Lambda = \Lambda(B,V) := \frac{1}{4}\left(\beta+\omega_1 +\sqrt{(\beta+\omega_1)^2+2\, \omega_2} \, \right)^2\, ,
\end{equation}
where $\beta, \omega_1$ and $\omega_2$ are non-negative constants which depend, in a weak sense, on the spacial asymptotics of $B$ and $V$. We refer to 
 Assumption \ref{ass-bounded infinity} and equation \eqref{def-beta-omega} for a full definition of $\beta$ and $\omega_j$. If $B$ and $V$ are regular enough at infinity, then the values of $\beta$ and $\omega_j$ are determined from their pointwise asymptotics. Indeed, splitting the potential into a sum of its short-range and long-range component; $V=V^s+V^\ell$, we find

 $$
 \beta\le  \, \limsup_{|x|\to\infty} |\wti B(x)|, 
\quad \omega_1 \, \le  \, \limsup_{|x|\to\infty} |x\, V^s(x)|_{\C^2}, \quad\text{and }\ \  \omega_2\, \le \, 
\limsup_{|x|\to\infty}\, | (x\cdot \nabla V^\ell(x))_+ |_{\C^2}.
 $$
see Lemma \ref{lem-pointwise}.

\begin{remark}
 It is illustrative to compare Theorem \ref{thm-abs}  with classical results on the absence of positive eigenvalues of non-magnetic Schr\"odinger operators \cite{ag,kato, s2}. 
If $B=0$, then by choosing $V^s=V$ and $V^\ell=0$ we obtain $\Lambda=\omega_1^2$ which generalizes the result of Kato \cite{kato}. On the other hand, by choosing $V^s$ such that $V^s(x) = o(|x|^{-1}),$ 
 and setting $V^\ell=V-V^s$ we get $\Lambda=\omega_2/2$, and recover thus the results of Agmon \cite{ag} and Simon \cite{s2}.
\end{remark}

To prove Theorem \ref{thm-abs} we adapt a version of the quadratic form method of \cite{ahk}, which in turn is  inspired by the approach  invented by Froese and Herbst for non-magnetic Schr\"odinger operators in \cite{fh,fhhh}.  However, due to the spinor structure of the operator $H_{A,V}$ and of its wave-functions, the technique of \cite{ahk,fh,fhhh} cannot be applied directly. The problem is that the operator-valued matrix $H_{A,V}$ is, contrary to the two-dimensional case, non-diagonal. 
Consequently, a direct application of the above mentioned technique, developed for scalar magnetic operators, is not feasible. It is therefore necessary to implement the fundamental 
ingredients of  \cite{ahk} in such a way that the spinor structure of  $H_{A,V}$  be taken into account.
To do so we make use of multiplication and commutation relations for the Pauli matrices and of their convenient interplay with the Poincar\'e gauge for the vector potential $A$. This is yet another example of the importance of choosing a gauge which suits best the problem in question. In our case the choice of the Poincar\'e gauge, together with the properties of the Pauli matrices, allows us to prove a matrix-valued versions of the virial-type identities for the weighted commutator between $H_{A,V}$ and the generator of dilations,  see equations  \eqref{eq-psiF-2} and \eqref{eq-energy boost-2}. With the help of these identities we then show that any eigenfunction 
of $H_{A,V}$ with eigenvalue larger than $\Lambda$ must identically vanish. We would like to point out that although the identities \eqref{eq-psiF-2} and \eqref{eq-energy boost-2} are identical to their scalar counterparts obtained in \cite{ahk}, due to the spinor structure of the problem under consideration their derivation is essentially different.

The paper is organized in the following way. In the first two sections we collect  necessary prerequisites and state our hypotheses. In Sections \ref{sec-aux} and \ref{sec-comm} we prove some preliminary results concerning dilations and commutator properties of $H_{A,V}$. The main result is stated and proved in Section \ref{sec-abs}.
In Section \ref{sec-example} we construct an example which shows that the critical energy $\Lambda(B,V)$ given by \eqref{edge} is sharp. As a consequence of Theorem \ref{thm-abs} we also establish sufficient conditions 
for the absence of  embedded eigenvalues of the magnetic Dirac operator, see Theorem \ref{thm-dirac} and Corollary \ref{cor-dirac}.
In Appendinx \ref{sec-pointwise} we show that all the hypothesis stated in Section \ref{sec-ass} are satisfied under some mind pointwise conditions on $B$ and $V$.

\section{\bf Prerequisites} \label{sec-prelim}
\subsection{Basic setup}
We identify the magnetic field with 
the vector--field $B:\R^3\to\R^3$ with components 
$(B_1,B_2, B_3) $. A vector potential is a vector field $A:\R^3\to \R^3$ which generates 
a magnetic fields via $B=\curl A$, in the distributional sense.   
We recall the well-known Pauli matrices $\sigma_j :\C^2\to \C^2$;
\begin{equation} \label{sigma-j}
\sigma_1 = \left( \begin{array}{cc}
0  & 1 \\
1 & 0
\end{array} \right), 
\qquad 
\sigma_2 = \left( \begin{array}{cc}
0  & -i \\
i & 0
\end{array} \right), 
\qquad 
\sigma_3 = \left( \begin{array}{cc}
1  & 0 \\
0 & -1
\end{array} \right) .
\end{equation}
In what follows we use the shorthand 
\begin{equation} 
z \cdot \sigma = \sum_{j=1}^3\,  z_j \, \sigma_j \qquad z\in\C^3. 
\end{equation}
The Pauli matrices satisfy the following  multiplication and commutation relations,
\begin{align}
\sigma_{j}\,  \sigma_k &= \delta_{j k} \id +i \sum_{m=1}^3 \eps_{j km}\, \sigma_m                 \label{pauli-re-1} \\
[ \sigma_{j}, \sigma_k] &= 2 i \sum_{m=1}^3 \eps_{j km}\,  \sigma_m             \label{pauli-re-2} \, .
\end{align}
Here $\id$ is the unite $2\times 2$ matrix, and $ \eps_{jkm}$ denotes the Levi-Civita permutation symbol. In particular, $\sigma_j^2= \id $ for $j=1,2,3$.

Given a magnetic field $B$ and a point $w\in\R^3$ let $\wti{B}_{w}(x)\coloneqq B(x+w)[x]$. More precisely, $\wti{B}_{w}$ is a vector--field on $\R^3$ defined by
\begin{align} \label{B-tilde}
\wti{B}_{w}(x+w & )= B(x+w)\wedge x	\, .
\end{align}
Making use of translations, we will often  assume $w=0$, in which case we will simply write $\wti{B}$.

\subsection{Notation} \label{ssec-notation} 
\noindent If $A\in L^2_{\text{loc}}(\R^3,\R^3)$ is a magnetic vector potential, the the magnetic Sobolev space is defined by 
\begin{equation} \label{wB}
\h^1(\R^3, \C^2)\coloneqq \D(P-A) = \Big\{ \varphi\in L^2(\R^3,\C^2)\, : \ (P-A)\, \varphi \in L^2(\R^3,\C^2) \Big\},
\end{equation}
equipped with the graph norm 
\begin{equation}
\| u\|_{\h^1} =  \Big( \|(P-A) u\|_{L^2(\R^3,\C^2) }^2  +  \| u\|_{L^2(\R^3,\C^2)}^2 \Big)^{1/2}\, .
\end{equation}
The corresponding scalar Sobolev space will be denoted by
\begin{equation*}
\h^1(\R^3) = \big\{ u\in L^2(\R^3)\, : \ (P-A)\, u \in L^2(\R^3) \big\} .
\end{equation*}

\smallskip

\noindent Given a set $M$ and two functions $f_1,\, f_2:M\to\R$, we write $f_1(x) \lesssim f_2(x)$ if there exists a numerical constant $c$ such that $f_1(x) \leq c\, f_2(x)$ for all $x\in M$. The symbol $f_1(x) \gtrsim f_2(x)$ is defined analogously. Moreover, we use the notation 
$$
f_1(x) \sim  f_2(x)  \quad \Leftrightarrow \quad f_1(x) \lesssim f_2(x) \ \wedge \ f_2(x) \lesssim f_1(x),
$$
and
\begin{equation} 
\limsup_{|x|\to \infty} f(x) = L \quad \Leftrightarrow \quad \lim_{r\to\infty} \,  \esssup_{|x|\geq r} f(x) = L,
\end{equation}
and similarly for $\liminf_{|x|\to \infty} f(x)$.
We will use $\partial_j=\frac{\partial}{\partial x_j}$ for the usual partial derivatives in the weak sense, i.e., as distributions.  

The scalar product on a Hilbert space $\mathscr{H}$ will be denoted  by $\LL\, \cdot\, , \, \cdot \RR_\mathscr{H}$. If $\mathscr{H}=L^2(\R^3,\C^2)$, we omit the subscript and write  
$$
 \LL \varphi , \, \psi \RR_{L^2(\R^3,\C^2)} =  \LL \varphi , \, \psi \RR \, , \qquad \varphi , \psi \in L^2(\R^3,\C^2). 
$$

Accordingly,  for any $\varphi\in L^r(\R^3,\C^2)$ with $1\le r\le \infty$ we will use the shorthand 
$$
\|\varphi\|_r : =\|\varphi\|_{L^r(\R^3, \C^2)}
$$
for the $L^r$-norm of $\varphi$. By the symbol 
$$
\U_R(x) = \{ y\in \R^3\, : \, |x-y| < R\}
$$
we denote the ball of radius $R$ centered at a point $x\in\R^3$. If $x=0$, we abbreviate 
$
\U_R = \U_R(0) . 
$

Given a Hermitian matrix valued function $\R^3\ni x\mapsto M(x):\C^2\to\C^2$, we denote by $\lambda(x)$ and $\mu(x)$ its eigenvalues. The norm of $M$  is then equal to 
$$
| M(x)|_{\C^2} = \max\big\{ \,|\lambda(x)| , \, |\mu(x)|\, \big\} .
$$
Accordingly we define
\begin{equation} \label{pos-part}
| M(x)_+|_{\C^2} = \max\big\{ \,\lambda(x)_+\,  , \, \mu(x)_+\, \big\} .
\end{equation}

\noindent {\bf Convention:} In the sequel we will use latin letters for functions with values in $\C$, and greek letters for functions with values in $\C^2$. In particular, we will often identify a spinor $\varphi$ with two scalar fields as follows;
\begin{equation} \label{phi-u-v}
\varphi =  \begin{pmatrix}
u \\
v  
\end{pmatrix} . \\[3pt]
\end{equation}
Throughout the paper we will often make use of the polarisation identity which, for the reader's convenience, we now briefly recall; given a sesquilinear form $s$ on a Hilbert space $\mathscr{H}$, and any $\varphi,\psi\in \mathscr{H}$, we have 
\begin{equation} \label{polarisation}
s(\varphi,\psi)  = \frac 14 \Big [ s(\varphi+\psi, \varphi+\psi) -s(\varphi-\psi, \varphi-\psi)  +is(\varphi-i\psi, \varphi-i\psi)  -i s(\varphi+i\psi, \varphi+i\psi) \Big ]. 
\end{equation}


\subsection{The Poincar\'e gauge} 
\noindent  For a given magnetic field $B$ and a point $w\in \R^3$ we define the vector field  $\wti{B}_{w}$ by equation \eqref{B-tilde}, and
put 
\begin{equation}  \label{eq-a}
A_{w}(x) \coloneqq \int_0^1 \wti{B}_w(tx)\, dt  \, , 
\end{equation}
which is the vector potential in the Poincar\'e gauge. Using translations, it is no loss of generality to assume $w=0$, in which case we will simply write $A$ for the vector potential given by \eqref{eq-a}. 
Note that when $w=0$, then $A$ given by \eqref{eq-a} satisfies 
\begin{equation} \label{poincare}
x\cdot A(x) =0 \qquad \forall\ x\in\R^3 .
\end{equation} 
It is easy to see that for $A$ given by \eqref{eq-a} one has $A\in L^2_{\text{loc} }(\R^3,\R^3)$ for bounded magnetic fields $B$ and this extends to a large class of singular magnetic fields, see  \cite[Lem.~2.9]{ahk}. Except otherwise stated, we will always use the Poincar\'e gauge.

\section{\bf Hypotheses} \label{sec-ass}

In this section we formulate general sufficient conditions on $B$ and $V$ under which our main result, Theorem \ref{thm-abs}, holds true.  In  Appendix \ref{sec-pointwise} we will show that all these conditions are satisfied under rather mild  assumptions on $B$ and $V$, see in particular Lemma \ref{lem-pointwise}, Proposition \ref{prop-sufficient} and Proposition \ref{prop-Lp-unif}.

\begin{assumption} \label{ass-V-hermit}
The matrix valued function $V:\R^3\to M(2,\C)$ is Hermitian. i.e.
\begin{equation} \label{V-hermit}
\big(V(x)\big)_{jk} = \overline{\big(V(x)\big)}_{kj} \qquad 	\forall\, x\in\R^3, \quad \forall\, j,k=1,2.
\end{equation} 
 If the potential is split as $V=V^s+V^\ell$, then $V^s$ and $V^\ell$ also satisfy \eqref{V-hermit}.
\end{assumption}

\begin{remark} 
Similarly as in the scalar non-magnetic case, see in particular \cite[Thm.~2.1]{fhhh}, our results could be extended to all matrix valued $V$, possibly non-Hermitian, for which the associated Pauli operator $H_{A,V}$ has real spectrum. For the sake of brevity, we will stick to  Assumption \ref{ass-V-hermit} throughout the paper.  
\end{remark}

 \begin{assumption}\label{ass-B-mild-int} The magnetic field $B$ is such that for some $w\in\R^3$ 
 \begin{align} \label{B-int-cond}
 	\R^3\ni x\, \mapsto\,  |x-w|^{-1}\, \log_+^2\Big( \frac{R}{|x-w|} \Big)\ |\wti{B}_{w}(x)|^2 \in L^1_{\text{loc}}(\R^3)
 \end{align}	
 for all $R>0$. 
 \end{assumption}

 \noindent We have already pointed out that without loss of generality we may assume 
 $w=0$. In view of \cite[Lem.~2.9]{ahk} condition \eqref{B-int-cond} assures that the corresponding vector potential in the Poincar\'e gauge is locally square integrable. The latter property is essential in order to define the Pauli operator through the associated quadratic form.

\subsection{Global relative bounds}
	
\begin{assumption} \label{ass-B-rel-bounded}
The scalar fields $|B|^2$ and $|\wti{B}|^{ 2} $ are relatively form bounded w.r.t.~$(P-A)^2$, where $A$ is the Poincar\'e gauge vector potential corresponding to $B$. That is, 
\begin{equation} \label{ass-B1-eq}
\| |B|  \varphi\|_2^2\, +  \||\wti{B}|  \varphi\|_2^2\ \lesssim \ \|(P-A)\varphi\|_2^2  +  \| \varphi\|_2^2 \qquad \forall\, \varphi\in \D(P-A). 
\end{equation}    
 \end{assumption}

\noindent Here we abuse the notation and use the same symbol $P-A$ for the operator in $L^2(\R^3)$ as well as for the operator in $L^2(\R^3,\C^2)$ acting as $\id (P-A)$.

\begin{assumption} \label{ass-V-form small}
The potential $V$ is relatively form small w.r.t.~$(P-A)^2$, that is,  
there exist constants $\alpha_0 < 1$ and $C$ such that 
\begin{equation} \label{ass-V0-eq}
 |\La \varphi, \, V \varphi \Ra |  \leq  \alpha_0\,  \|(P-A) \varphi\|_2^2  + C \| \varphi\|_2^2 \qquad \forall\, \varphi\in  \D(P-A).
\end{equation} 
\end{assumption}

\noindent In order to control the 
virial $x\cdot\nabla V$, we decompose the potential as $V=V^s+V^\ell$.
The splitting $V=V^s+V^\ell$ is arbitrary, as long as the conditions below are satisfied.

\subsection{ Behaviour at infinity} 
Below we quantify the notions of boundedness and vanishing at infinity w.r.t.~$(P-A)^2$.

\begin{definition}[\bf Boundedness at infinity] 
\label{def-bounded infinity}
A potential $W$ is bounded from above at infinity 
with respect to $(P-A)^2$ if for some $R_0>0$ its 
quadratic form domain contains 
all $\varphi\in\D(P-A)$ with $\supp(\varphi)\in \U_{R_0}^{\,c}$ 
and for $R\ge R_0$ there exist positive  $\alpha_R, \gamma_R$ 
with $\lim_{R\to\infty}\alpha_R=0$ and 
$\liminf_{\R\to\infty}\gamma_R<\infty$ such that 
	\begin{align}\label{eq-bounded infinity}
		\La \varphi, W\varphi \Ra  \le 
		\alpha_R \|(P-A)\varphi\|_2^2 + \gamma_R\|\varphi\|_2^2 \quad \text{for all } \varphi\in\D(P-A) 
		\text{ with } \supp(\varphi)\subset \U_R^c
	\end{align}
	 By monotonicity  we may assume, without loss of 
	generality,  that  $\alpha_R$ and $\gamma_R$ are 
	decreasing in $R\ge R_0$.
	\end{definition}

\smallskip

\begin{assumption} \label{ass-V-vanishing infinity} The positive part of the potential $V$ vanishes  
at infinity w.r.t.~$(P-A)^2$ in the following sense: there exist positive  $\alpha_R, \gamma_R$ with 
	$\alpha_R, \gamma_R \to 0$ as $R\to\infty$ such that 
	\begin{align}
	\la \varphi, V\varphi \ra_+  \le 
		\alpha_R \|(P-A)\varphi\|_2^2 + \gamma_R\|\varphi\|_2 \quad \text{for all } \varphi\in\D(P-A) 
		\text{ with } \supp(\varphi)\subset \U_R^c\, .
	\end{align}  
Moreover, if we split $V=V^s+V^\ell$, then also the positive parts of $V^s$ and $V^\ell$ vanishe at infinity in the sense defined above.
By monotonicity  we may assume, without loss of 
generality,  that  $\alpha_R$ and $\gamma_R$ are decreasing in $R\ge R_0$.
\end{assumption}

\smallskip

\begin{assumption} \label{ass-V-bounded-infinity} The potential $V$ is bounded 
at infinity w.r.t.~$(P-A)^2$ in the sense of Definition \ref{def-bounded infinity}. Moreover, if we split $V=V^s+V^\ell$, then also $V^s$  is bounded at infinity w.r.t~$(P-A)^2$ in the sense of Definition \ref{def-bounded infinity}.
\end{assumption}

\begin{assumption} \label{ass-bounded infinity} 
There exist positive sequences $(\eps_j)_{j}, (\beta_j)_{j}$ and $(R_j)_{j}$ with $\eps_j \to 0$ and $R_j\to\infty$ as $j\to\infty$, such that for all $\varphi\in \D(P-A)$ with ${\rm supp}(\varphi) \subset \U_j^{\, c}=\{x\in\R^3: |x|\ge R_j\}$ 
\begin{align} \label{ass-B-inft}
\| | \wti{B}|  \varphi\|_2^2  \  & \leq \   \eps_j\,  \|(P-A) \varphi\|_2^2  + \beta_j^2 \, \| \varphi\|_2^2 
\end{align} 
For the decomposition $V=V^s+V^\ell$ of the potential, we also assume that there exist positive sequences 
$(\omega_{1,j})_{j}$  and $(\omega_{2,j})_{j}$ 
such that for all $\varphi\in \D(P-A)$ with ${\rm supp}(\varphi) \subset \U_j^{\, c}$
\begin{align} 
\| x\,  V^s \varphi\|_2^2   \  & \leq \   \eps_j\,  \|(P-A) \varphi\|_2^2  + \omega_{1,j}^{\, 2} \, \| \varphi\|_2^2 \\
\La \varphi,  x\cdot \nabla V^\ell \varphi \Ra\ &  \leq \   \eps_j\,  \|(P-A) \varphi\|_2^2  + \omega_{2,j}\,  \| \varphi\|_2^2  \label{ass-x nablaV_2-inft}
\end{align}  
\end{assumption}

\noindent By monotonicity  we may assume, without loss of 
generality,  that the sequences $\beta_j$, $\omega_{1,j}$, 
and $\omega_{2,j}$ in Assumptions \ref{ass-bounded infinity} are decreasing. We define
\begin{equation} \label{def-beta-omega}
\beta \coloneqq \lim_{j\to \infty} \beta_j \, , 
\qquad \omega_k\coloneqq \lim_{j\to \infty} \omega_{k,j} \,  , \quad k=1,2.
\end{equation}

\subsection{Unique continuation at infinity.} 
\noindent For a unique continuation type argument at infinity, we also need a quantitative version of relative form boundedness.  

\begin{assumption} \label{ass-bounded unique continuation}
 If $V=V^s+V^\ell$, then we assume  that for all $\varphi\in\D(P-A)$
 \begin{align} 
  \| |\wti{B} |  \varphi\|_2^2 +\|xV^s\varphi\|_2^2  \
  	&  \leq \   \frac{\alpha_1^2}{4}\,  \|(P-A) \varphi\|_2^2  + C \| \varphi\|_2^2 , \label{ass-B2-eq}\\
  \La \varphi,    x\cdot \nabla V^\ell\,  \varphi \Ra\  
    &  \leq \   \alpha_2\,  \|(P-A) \varphi\|_2^2  + C\| \varphi\|_2^2 , \label{ass-V1-eq} \\
  |\La \varphi,  V^s\, \varphi \Ra|\  & \leq \   \alpha_3\,  \|(P-A) \varphi\|_2^2  + C \| \varphi\|_2^2   \label{ass-xV-quantitative}
 \end{align} 
 for some $ C>0$ and $\alpha_j$ such that
\begin{equation} \label{alphas}
  \alpha_1 + \alpha_2 +3\alpha_3< 1. 
\end{equation} 
\end{assumption} 

\begin{remark}
By the diamagnetic inequality
 \begin{equation}  \label{diamagnetic-inequality}
\big | P |\varphi | \big | \leq \big | (P-A)\, \varphi \big | \quad \text{a.e.} \qquad \text{for all } \varphi\in\D(P-A), 
 \end{equation}
 see e.g.~\cite{kato-2}, it suffices to verify the conditions of Assumptions \ref{ass-B-rel-bounded}-\ref{ass-bounded unique continuation} with $(P-A)$ replaced by $P$.
\end{remark}

\smallskip



\section{\bf Preliminary results} \label{sec-aux}
\noindent In this section we collect several technical results which will be needed later. 

\begin{lemma} \label{lem-1}
Let $A\in L_{\rm loc}^2(\R^3)$ and let $B= \curl\, A$. Suppose moreover that $B$ satisfy Assumption \ref{ass-B-rel-bounded}. Then 
$$
\LL \sigma\cdot (P-A)\varphi,\,  \sigma\cdot (P-A)\varphi \RR= \| (P-A)\varphi\|^2_2 +\LL \varphi , \, \sigma \cdot B\, \varphi \RR   \qquad \forall\ \varphi \in \h^1(\R^3, \C^2).
$$
\end{lemma}
\begin{proof}
The claim follows by a direct calculation from \eqref{pauli-re-1} and \eqref{pauli-re-2}.
\end{proof}

\begin{lemma} \label{lem-2a}
Let $B$ satisfy Assumption \ref{ass-B-rel-bounded}. 
For any $\eta >0$ there exists $C_\eta\in\R$ such that 
\begin{equation} \label{sigma-bound}
\|  (P-A)\varphi\|_2^2 \, \leq \, (1+\eta) \|\sigma\cdot (P-A)\varphi\|_2^2 +C_\eta \|\varphi\|_2^2 
\end{equation}
holds for all $\varphi\in\D(P-A)$.
\end{lemma}

\begin{proof}
Let $\varphi\in\D(P-A)$. A short calculation shows that 
\begin{equation} \label{sigma-w}
\|\sigma\cdot w\, \varphi\|_2 =\| |w| \, \varphi\|_2\qquad \forall\, w\in\C^3.
\end{equation}
Hence by \eqref{ass-B1-eq}, Lemma \ref{lem-1} and Cauchy-Schwarz inequality,
\begin{align*}
\|  (P-A)\varphi\|_2^2 & = \|\sigma\cdot (P-A)\varphi\|_2^2 - \LL \varphi , \, \sigma \cdot B\, \varphi \RR \leq   \|\sigma\cdot (P-A)\varphi\|_2^2 +\| |B| \, \varphi\|_2  \|\varphi\|_2 \\
& \leq  \|\sigma\cdot (P-A)\varphi\|_2^2 +\eps \| (P-A) \, \varphi\|_2^2 +C_\eps \|\varphi\|_2^2
\end{align*}
for any $0<\eps<1$ and some $C_\eps$, independent of $\varphi$. Inequality \eqref{sigma-bound} now follows upon setting $\frac{1}{1-\eps} = 1+\eta$. 
\end{proof}

\smallskip

\noindent An immediate consequence of Lemma \ref{lem-2a} is the following 

\begin{corollary} \label{cor-uc}
Let $B$ and $V$ satisfy Assumptions \ref{ass-B-rel-bounded},  \ref{ass-V-bounded-infinity} and \ref{ass-bounded unique continuation}. Then for all $\varphi\in\D(P-A)$, 
 \begin{align*} 
  |\La \varphi, \, V \varphi \Ra |  &  \leq \    \alpha_0\,  \|(P-A) \varphi\|_2^2  + C_0\| \varphi\|_2^2 \\
  \| \wti{B}  \varphi\|_2^2 +\|xV^s\varphi\|_2^2  
  	&  \leq \  \displaystyle \frac{\alpha_1^2}{4} \,  \|(P-A) \varphi\|_2^2  + C_1 \| \varphi\|_2^2 ,\\
  \La \varphi,    x\cdot \nabla V^\ell\,  \varphi \Ra\  
    &  \leq \   \alpha_2\,  \|(P-A) \varphi\|_2^2  + C_2 \| \varphi\|_2^2 , \label{ass-V1-eq} \\
  |\La \varphi,  V^s\, \varphi \Ra|\  & \leq \   \alpha_3\,  \|(P-A) u\|_2^2  + C_3 \| \varphi\|_2^2  
 \end{align*} 
where $\alpha_0<1$,  and $\alpha_j,\, j=1,2,3$  satisfy \eqref{alphas}.
\end{corollary}

\noindent Another consequence of Lemma \ref{lem-2a} is the identity
\begin{equation}  \label{eq-domains}
\D(P-A) = \D(\sigma\cdot (P-A))\, ,
\end{equation}
which holds whenever Assumption \ref{ass-B-rel-bounded} is satisfied. This allows us to define the sesquilinear form 
\begin{equation} \label{q0}
Q_{A,0} (\varphi, \psi) = \LL \sigma\cdot (P-A)\varphi,\,  \sigma\cdot (P-A)\psi \RR, \qquad \varphi, \psi \in \D(P-A) .
\end{equation}
By standard arguments one verifies that the quadratic form $Q_{A,0}(\varphi, \varphi)$ is closed. 
In view of Lemma \ref{lem-2a}  and Assumption \ref{ass-V-form small} the quadratic form associated to 
\begin{equation} \label{q}
Q_{A,V}(\varphi,  \psi) = Q_{A,0} (\varphi, \psi)  + \LL \varphi, V \psi\RR , \qquad \varphi,\psi \in \D(P-A) 
\end{equation}
is then closed as well.  Now we can define the Hamiltonians $H_{A,0}$ and $H_{A,V}$ as the unique self-adjoint operators associated to $Q_{A,0}$ and $Q_{A,V}$ respectively. 
	
For the next result we need to introduce some additional notation. Given a vector field $v: \R^3\to \R^3$ we define the operator 
\begin{equation} \label{op-D}
D_v = \frac 12 \big(v\cdot P +P\cdot v), \qquad D: = D_x \quad \text{if} \quad v= x.
\end{equation}

\begin{lemma} \label{lem-2}
Let $B$ satisfy Assumption \ref{ass-B-mild-int}. Let $g, F\in C^1(\R^3)$ $g, F\in C^1(\R^3; \R)$ be radial functions  such that $\nabla F = x g$, and such that $x\cdot \nabla g$ and $|x| g$ are bounded. Put $v = \nabla F$. Then 
\begin{equation} \label{domains-q}
  \D(P-A) \subset {\D(D_v) = \D(g D)}
\end{equation}
\end{lemma} 

\begin{proof}
In the sense of distributions,
$$
2 D_v = gx \cdot P + P\cdot (g x) = gx\cdot P +g P\cdot x -ix\cdot \nabla g = 2g D -ix\cdot \nabla g.
$$
So if $\varphi \in \D(g D)$ and $x\cdot\nabla g$ is bounded, then 
$$
D_v\varphi  = g D \varphi -i x\cdot \nabla g \varphi \in L^2(\R^3,\C^2).
$$
Hence $\varphi\in \D(D_v)$. Conversely, if  $\varphi \in \D(D_v)$ and $x\cdot\nabla g$ is bounded, then  $\varphi \in \D(g D)$. 
This proves the equality   $\D(D_v) = \D(g D)$. 
 Moreover, since $x\cdot A(x) =0$, 
$$
2 g D = g (x\cdot P+ P\cdot x)= 2g x\cdot P -3i = 2g x\cdot (P-A) -3i.
$$
So if $\varphi\in \D(P-A)$, and $|x| g$ is bounded, then $g D \varphi \in L^2(\R^3,\C^2)$. Hence $  \D(P-A) \subset \D(g D) =\D(D_v)  $.
\end{proof}


\smallskip

\begin{lemma} \label{lem-3}
Under the assumptions of Lemma \ref{lem-2}, 
\begin{equation} \label{eq-lem-3}
{\rm Re} \LL (\sigma \cdot v)\ \varphi ,\, \sigma\cdot (P-A)\, \varphi \RR = \LL \varphi, \, D_v \id \varphi\RR \qquad \forall\ \varphi\in  \D(P-A).
\end{equation}
\end{lemma} 

\smallskip

\begin{remark}   
We note that thanks to Lemma \ref{lem-2}, the right hand side of \eqref{eq-lem-3} is well defined for all $\varphi\in  \D(P-A)$. Lemma \ref{lem-2} is also used implicitly in 
Lemma \ref{lem-a-bound} and in Proposition \ref{prop:magnetic-virial}.
\end{remark}

\begin{proof}[Proof of Lemma \ref{lem-3}]
Let $\varphi\in C_0^\infty(\R^3, \C^2)$. Then
$$
{\rm Re\, } \LL (\sigma \cdot v)\ \varphi ,\, \sigma\cdot (P-A)\, \varphi \RR = {\rm Re\, } \LL \varphi ,\,  (\sigma \cdot v) \sigma\cdot (P-A)\, \varphi \RR = 
  {\rm Re\, } \LL \varphi ,\,  (\sigma \cdot v) \sigma\cdot P\, \varphi \RR - {\rm Re\, } \LL \varphi ,\,  (\sigma \cdot v) \sigma \cdot A\, \varphi \RR.
$$
In view of \eqref{pauli-re-1},
$$
(\sigma \cdot v)  (\sigma \cdot A) = \sum_{j,k=1}^3 v_j A_k \sigma_j \sigma_k = (v\cdot A) \, \id +i \sum_{m=1}^3 \Big( \sum_{j,k=1}^3 \eps_{mjk}\,  v_j A_k \Big) \sigma_m = 
(v\cdot A) \, \id + i (v\wedge A)\cdot \sigma\, .
$$
Hence
$$
{\rm Re\, } \LL \varphi ,\,  (\sigma \cdot v) \sigma \cdot A\, \varphi \RR = \frac 12 \LL \varphi ,\, \big [ (\sigma \cdot v) (\sigma \cdot A) + (\sigma \cdot A) (\sigma \cdot v)  \,\big ]  \varphi \RR
= \LL \varphi, \, (v\cdot A) \, \id \varphi \RR,
$$
because $v\wedge A +A\wedge v=0$. But $v(x) = g(|x|) x$ by assumption, and $A$ is in the Poincar\'e gauge. So $v \cdot A=0$, see \eqref{poincare}. We thus have
\begin{equation}  \label{partial} 
{\rm Re\, } \LL (\sigma \cdot v)\ \varphi ,\, \sigma\cdot (P-A)\, \varphi \RR  = {\rm Re\, } \LL \, \varphi ,\,  (\sigma \cdot v) ( \sigma\cdot P)\, \varphi \RR =  \frac 12 \LL \varphi ,\, \big [ (\sigma \cdot v) (\sigma \cdot P) + (\sigma \cdot P) (\sigma \cdot v)  \,\big ]  \varphi \RR .
\end{equation} 
Now, using \eqref{pauli-re-1} we get
\begin{align*}
(\sigma \cdot v) (\sigma \cdot P) + (\sigma \cdot P) (\sigma \cdot v) &= \sum_{j,k=1}^3 (v_j P_k+P_j v_k) \, \sigma _j \sigma_k =  \sum_{j,k=1}^3 (v_j P_k+P_j v_k) \Big( \delta_{jk} \id +i
\sum_{m=1}^3 \eps_{jkm} \, \sigma_m\Big)\\
& = ( v\cdot P +P\cdot v)\, \id +i \!\sum_{j,k,m=1}^3 (v_j P_k +P_j v_k) \, \eps_{jkm}\, \sigma_m \\
& = ( v\cdot P +P\cdot v)\, \id +i\! \sum_{j,k,m=1}^3 (v_j P_k -P_k v_j) \, \eps_{jkm}\, \sigma_m \\
&= ( v\cdot P +P\cdot v)\, \id +\sum_{j,k,m=1}^3 (\partial_k v_j) \, \eps_{jkm}\, \sigma_m = ( v\cdot P +P\cdot v)\, \id  = 2 D_v \, \id,
\end{align*}
where we have used the identity
$$
\sum_{j,k,m=1}^3 (\partial_k v_j) \, \eps_{jkm}\, \sigma_m = (\curl\,  v)\cdot \sigma = (\curl\,  \nabla F)\cdot \sigma =0.
$$
Summing up, we have 
\begin{equation}
{\rm Re} \LL (\sigma \cdot v)\ \varphi ,\, \sigma\cdot (P-A)\, \varphi \RR = \LL \varphi, \, D_v \id\, \varphi\RR \qquad \forall\ \varphi\in  C_0^\infty(\R^3, \C^2).
\end{equation}
Since $v= \nabla F$ is bounded, this identity extends by density to all $\varphi\in  \D(P-A)$.
\end{proof}

\section{\bf Dilations and the commutator } \label{sec-comm}

In this section we will define the commutator $[H_{A,V},\,  D]$ in the sense of quadratic form and derive a matrix-valued version of the weighted virial identities. The latter are our main technical tools in the proof of absence of positive eigenvalues. In some places we make use of technical results obtained in \cite{ahk}.


\subsection{Dilations} 

For $t\in\R $ define the unitary dilation operator $U_t: L^2(\R^3, \C^2)\to L^2(\R^3, \C^2)$ by 
\begin{equation}  \label{eq:U}
(U_t f)(x) = e^{\frac{3t}{2}} f(e^t x) \, \quad x\in\R^3.
\end{equation} 
Then $U_t= e^{i t D}$ on $L^2(\R^3, \C^2)$. Let
\begin{equation} \label{gt-def}
i G_t = \frac{U_t - U_{-t}}{2t} \qquad t\in\R.   
\end{equation}
It is easily seen that $G_t$ is bounded and symmetric on $L^2(\R^3, \C^2)$. We will use it to approximate the operator $D$ in the limit $t\to 0$.   

\smallskip

\noindent As in \cite{ahk} we define the commutator of $H$ and $D$ by
\begin{equation} \label{eq:comutator-form-1}
\La \varphi, i \, [H_{A,V},\,  D] \, \varphi \Ra\coloneqq \lim_{t\to 0} \la \varphi , [H_{A,V}, iG_t] \, \varphi  \ra \coloneqq  2 \lim_{t\to 0} \re Q_{A,V}(\varphi , i G_t\, \varphi )\, , 
\end{equation} 
provided the limit on the right hand side exists.  Recall that $\D(Q_{A,V})$ is invariant under dilations, see \cite[Prop.~3.3]{ahk}, hence $Q_{A,V}(\varphi , i G_t\, \varphi )$ is well defined for any $t\neq 0$.

\begin{lemma} \label{lem-a-bound}
Let $B$ satisfy Assumption \ref{ass-B-mild-int}. Then 
\begin{equation} \label{com-1}
\La \varphi, i \, [H_{A,0},\,  D] \, \varphi \Ra  = 2 \lim_{t\to 0} \re Q_{A,0}(\varphi, i G_t\, \varphi)\,  = 2 \|\sigma\cdot(P-A) \varphi \|_2^2 + 2 \re \LL \sigma\cdot \wti B\, \varphi, \, \sigma\cdot(P-A)\, \varphi\RR
\end{equation} 
for all $\varphi\in\D(P-A)$.
\end{lemma}

\begin{proof}
Let $\varphi$ be given by \eqref{phi-u-v} with $u,v \in\h^1(\R^3)$. A short calculation shows that 
\begin{equation*} 
(P_j-A_j) U_t\, u = e^tU_t (P_j-A_j)\, u + X^j_t\, u \quad \text{with} \quad X^j_t  =  U_t (e^t\, A_j-A_j (e^{-t}\, \cdot) )
\end{equation*} 
for any $j=1,2,3$. 
Hence for any $w\in L^2(\R^3)$, 
\begin{equation} 
\LL w, (P_j-A_j) (U_t-U_{-t}) u\RR = \LL w, (e^t U_t-e^{-t} U_{-t}) (P_j-A_j) u \RR + \LL w, (X^j_t-X^j_{-t}) u\RR\, .
\end{equation} 
Since 
$$
\lim_{t\to 0} \, t^{-1}\, X^j_{\pm t}\,  u  = \pm\wti{B}_j\,  u \quad \text{in} \ \ L^2(\R^3),
$$
see \cite[Prop.~3.6]{ahk}, we deduce from \eqref{gt-def} that
$$
2 \re  \lim_{t\to 0} \re \LL  w,\,  (P_j-A_j)  i G_t\, u \RR_{ L^2(\R^3)}= 2 \re\LL w,  (P_j-A_j) u\RR  + 2 \re \LL w, \wti{B}_j\, u\RR_{L^2(\R^3)}.
$$
After an elementary, but lengthly calculation we then obtain
\begin{align*}
\La \varphi, i \, [H_{A,0},\,  D] \, \varphi \Ra &= 2 \lim_{t\to 0} \re \LL \sigma\cdot (P-A) \varphi,\,  \sigma\cdot (P-A)  i G_t\, \varphi \RR \\
&= 2 \|\sigma\cdot(P-A) \varphi \|_2^2 + 2 \re \LL \sigma\cdot \wti B\, \varphi, \, \sigma\cdot(P-A)\, \varphi\RR\, ,
\end{align*}
as claimed.
\end{proof}

\begin{lemma} \label{lem:xgradv}
Let $B$ satisfy Assumption \ref{ass-B-mild-int} and let  $W:\R^3\to\R$ be a potential with form domain containing $\D(P-A)$, such that the distribution $x\cdot \nabla W$ extends to a quadratic form which is form bounded with respect to $(P-A)^2$. Then 
\begin{align}
	\label{eq-deriv-2}
2 \lim_{t\to 0} \, \la u,\,   W i G_t\,  v \ra_{L^2(\R^3)} = -\la u,  x\cdot \nabla W \, v\ra_{L^2(\R^3)} 
\end{align}
for all $u,v \in\h^1(\R^3)$. 
\end{lemma} 

\begin{proof}
By \cite[Lemma~3.7,  Eq. (3.32]{ahk} we have 
$$
2 \lim_{t\to 0} \, \la u,\,   W i G_t\,  u \ra_{L^2(\R^3)} = -\la u,  x\cdot \nabla W \, u\ra_{L^2(\R^3)} 
$$
The claim thus follows again from the polarisation identity \eqref{polarisation}.
\end{proof}

\subsection{The commutator} The following result provides a matrix-operator version of a magnetic virial theorem. 

\begin{proposition}\label{prop:magnetic-virial} 
Let $B$ and $V$ satisfy Assumptions \ref{ass-B-mild-int}-\ref{ass-V-form small}.  Suppose moreover that $x\cdot\nabla V$  is form bounded with respect to 
$(P-A)^2$.  
Then for all $\varphi\in \D (\sigma\cdot (P-A))$, the limit $\lim_{t\to0}\re\big(Q_{A,V}(\varphi, i G_t \varphi)\big) $ in \eqref{eq:comutator-form-1} exists. Moreover,
\begin{equation}  \label{eq-hd}
\La\varphi, [H_{A,V}, iD] \, \varphi \Ra  = 2 \|\sigma\cdot (P-A) \varphi\|_2^2  + 2 \re \La \sigma\cdot \wti{B}\, \varphi ,\, \sigma\cdot (P-A)\, \varphi\Ra- \La \varphi, x\cdot\nabla V\varphi \Ra \, .
\end{equation}
\end{proposition}

\begin{proof}
Let $\varphi\in \D (\sigma\cdot (P-A))$ be given by \eqref{phi-u-v}. In view of Lemma  \ref{lem-a-bound} it suffices to show that 
\begin{equation} \label{enough-1} 
\La\varphi, [V, iD] \, \varphi \Ra = 2 \lim_{t\to 0} \re\,  \la \varphi , V i G_t\, \varphi \ra\, = - \La \varphi, x\cdot\nabla V\varphi \Ra \, .
\end{equation}
Let $V_{jk}$ denote the matrix elements of $V$. By hypothesis of the proposition we have 
\begin{align} \label{x-V-bound}
\big |\la u,  x\cdot \nabla V_{11} \, u\ra_{L^2(\R^3)}  & +\la v,  x\cdot \nabla V_{22} \, v\ra_{L^2(\R^3)}  +\la u,  x\cdot \nabla V_{12} \, v\ra_{L^2(\R^3)}  +\la v,  x\cdot \nabla V_{12} \, u\ra_{L^2(\R^3)} \big| \nonumber \\
& \lesssim\, \|(P-A) u\|_{L^2(\R^3)}^2 +\|(P-A) v\|_{L^2(\R^3)}^2 +\| u\|_{L^2(\R^3)}^2+\|v\|_{L^2(\R^3)}^2 ,
\end{align} 
for all $u,v \in \h^1(\R^3)$. 
Applying the above inequality first with $v=0$ and then with $u=0$ shows that  $x\cdot \nabla V_{11}$ and $x\cdot \nabla V_{22}$ are  relatively form bounded with respect to 
$(P-A)^2$ in $L^2(\R^3)$. Hence if we return to  \eqref{x-V-bound} and put  $u=v$, then using the triangle inequality we deduce that also the quadratic form 
$$
\la u,  x\cdot \nabla V_{12} \, u\ra_{L^2(\R^3)} + \la u,  x\cdot \nabla V_{21} \, u\ra_{L^2(\R^3)}  = 2 {\rm Re} \la u,  x\cdot \nabla V_{12} \, u\ra_{L^2(\R^3)} 
$$
is relatively bounded with respect to $(P-A)^2$ in $L^2(\R^3)$. Equation \eqref{enough-1}, and hence the claim, thus follows from Lemma \ref{lem:xgradv} and  \eqref{polarisation}.
\end{proof}

\begin{remark}
For rigorous results on 
virial identities, which have a long history 
in mathematics and physics, we refer e.g.~to
\cite{weidmann} and \cite{alb}.
\end{remark}

\subsection{Exponentially weighted commutator}
\label{ssec-exp-w-com} 
The crucial ingredient for the proof of our main result, see Theorem \ref{thm-abs} below, is finding two different expressions for the weighted commutator $\la e^F\psi, [H_{A,V}, D]\, e^F\psi\ra$, when $F$ is a weight function and $\psi$ is a weak eigenfunction of $H_{A,V}$. This is provided by the following Lemma and by the subsequent equation \eqref{eq-energy boost-2}.

\begin{lemma} \label{lem-com-ef}
Let $B$ and $V$ satisfy Assumptions \ref{ass-B-mild-int}-\ref{ass-V-form small}.  
Assume that $x\cdot\nabla V$ is form bounded with respect to $(P-A)^2$. 
  Let $F \in C^2(\R^3;\R)$ be a
  bounded radial function, such that  $\nabla F =x  g$, and assume 
  that $g\ge0$ and  
  that the functions $\nabla(|\nabla F|^2)$, $(1+|\cdot|^2) g$,  $x\cdot\nabla g$ and 
  $(x\cdot\nabla)^2g$ are bounded. 
Let $\psi\in\D(P-A)$ be a weak eigenfunction of $H_{A,V}$, i.e., 
  $E\la \varphi,\psi\ra= Q_{A,V}(\varphi,\psi)$ for some $E\in\R $ and 
  all $\varphi\in\D(P-A)$. Then 
\begin{align}
\La\psi_F, i \, [H_{A,V}, D] \, \psi_F \Ra& = -4 \, \|\sqrt{g}\, D\, \psi_F\|_2^2 + \La\psi_F, \big((x\cdot \nabla)^2 g -x\cdot \nabla |\nabla F|^2  \big) \psi_F \Ra    \, ,\label{eq-psiF-2} 
\end{align} 
where $\psi_F = e^F\, \psi$.
\end{lemma}

\begin{proof}
Note that in the sense of quadratic forms 
\begin{equation} \label{fhf}
e^{F}\, H_{A,V}\, e^{-F} = H_{A,V} +i \big[ (\sigma \cdot \nabla F) \sigma \cdot (P-A) +(\sigma \cdot (P-A) )\, \sigma \cdot \nabla F\big] - |\nabla F|^2\, .
\end{equation}
Hence 
\begin{align}
\LL e^F \psi, [H_{A,V}, i D]\,  e^F \psi\RR & = 2\re\, \LL H_{A,V}\, e^F \psi, iD e^F \psi\RR = 2\re\, \LL e^F H_{A,V}\,  e^{-F} e^F \psi, iD\, e^F \psi\RR\nonumber \\
& \quad -2\re \LL \big( (\sigma \cdot \nabla F) \sigma \cdot (P-A) +(\sigma \cdot (P-A) )\, \sigma \cdot \nabla F\big)  \, \psi_F,  D\,  \psi_F\RR\nonumber \\
& \quad +2\re \, \LL |\nabla F|^2 \, \psi_F, \, i D\,  \psi_F \RR\nonumber \\
&= 2\re \LL \big( (\sigma \cdot \nabla F) \sigma \cdot (P-A) +(\sigma \cdot (P-A) )\, \sigma \cdot \nabla F\big)   \psi_F,  D\,  \psi_F\RR\nonumber \\
&\quad  - \LL   \psi_F, x\cdot \nabla |\nabla F|^2\, \psi_F \RR,  \label{com-fhf}
\end{align}
where we have used the fact that
$$
 \re \LL e^F H_{A,V}\,  e^{-F} e^F \psi, iD\,  e^F \psi\RR  =   E\,  \re\LL     \psi_F, iD\,    \psi_F\RR =0.
$$
Now since $F$ is radial and $A$ is in the Poincar\'e gauage \eqref{poincare}, it follows from \eqref{pauli-re-1} that
$$
(\sigma \cdot \nabla F) ( \sigma \cdot A) +(\sigma \cdot A )\, ( \sigma \cdot \nabla F) = 2 (\nabla F\cdot A) \, \id =0.
$$
On the other hand, still using  \eqref{pauli-re-1} we obtain
\begin{align*}
(\sigma \cdot \nabla F) (\sigma \cdot P) +(\sigma \cdot P )\, (\sigma \cdot \nabla) F &= (\nabla F \cdot P +P\cdot \nabla F)\, \id + i \sum_{j,k,m=1}^3 \big( \partial_j F\, P_k + P_j \, \pd_k F\big)\, 
\eps_{jkm}\, \sigma_m\\
& = (g x \cdot P +P\cdot x g)\, \id +i \sum_{j,k,m=1}^3 \big( \partial_j F\, P_k - P_k \, \pd_j F\big) \, \eps_{jkm}\, \sigma_m \\
& = (g D -ix \cdot\nabla g )\, \id +\sum_{m=1}^3  \Big( \sum_{j,k=1}^3  \pd_k \partial_j F\,  \eps_{jkm}\, \Big)  \sigma_m\, .
\end{align*}
Since $\sum_{j,k=1}^3  \pd_k \partial_j F\,  \eps_{jkm}=0$ for all $m=1,2,3$, the last equation in combination with \eqref{com-fhf} gives 
\begin{align*}
\LL  \psi_F, [H_{A,V}, i D]\,  \psi_F\RR &= -\re\,  \LL  (g D -ix \cdot\nabla g )\,  \psi_F, \,  D\,  \psi_F\RR - \LL   \psi_F, x\cdot \nabla |\nabla F|^2\,  \psi_F\RR \\
& = -4 \, \|\sqrt{g}\, D\, \psi_F\|_2^2 + \La\psi_F, \big((x\cdot \nabla)^2 g -x\cdot \nabla |\nabla F|^2  \big) \psi_F \Ra \, .
\end{align*}
\end{proof}

In view of the fact that $D_{\nabla F}$ is symmetric equation \eqref{fhf} and Lemma \ref{lem-3} imply 
  \begin{equation}\label{Fmp}
  Q_{A,V}(\varphi,\varphi) =
   Q_{A,V}(e^{-F} \varphi, e^{F} \varphi) + \La\nabla F \, \varphi, \nabla F\, \varphi \Ra
  \end{equation}
  for all $\varphi\in \D(P-A)$. By inserting $\varphi= \psi_F$ in the above equation, which is allowed because $\psi_F\in \D(P-A)$, we get   
  \begin{equation}\label{eq-energy boost}
		Q_{A,V}(\psi_F,\psi_F) = \| \sigma\cdot (P-A) \psi_F\|_2^2+ \LL \psi_F, \, V \psi_F\RR = \La \psi_F,(E+|\nabla F|^2)\psi_F  \Ra\, . \\[7pt]
  \end{equation}
A combination  with \eqref{eq-hd} thus gives 
\begin{equation} \label{eq-energy boost-2}
  \begin{aligned}
\La\psi_F, i \, [H_{A,V}, D] \, \psi_F \Ra &=  \La \psi_F,(E+|\nabla F|^2)\psi_F  \Ra + \| \sigma\cdot (P-A) \psi_F\|_2^2 +2 \re \La \sigma\cdot \wti{B}\,  \psi_F ,\, \sigma\cdot (P-A)\,  \psi_F\Ra \\
& \ \quad- \La  \psi_F, (V+x\cdot\nabla V) \psi_F\Ra.
  \end{aligned}
  \end{equation}

\begin{lemma} \label{lem-lowerb}
Let $B$ and $V$ satisfy Assumptions \ref{ass-B-mild-int}, \ref{ass-V-form small}, and \ref{ass-bounded unique continuation}. Assume that $\psi$ and $F$ satisfy conditions of Lemma \ref{lem-com-ef}. Then there exist constants $\kappa>0$ and $c_\kappa>0$ such that 
\begin{equation} \label{com-lowerb-2} 
\La\psi_F,  \, [H_{A,V},  i D] \, \psi_F \Ra\, \geq  \, \kappa\,  \La\psi_F, |\nabla F|^2\, \psi_F \Ra- c_\kappa \|\psi_F\|_2^2 \, .
\end{equation}
\end{lemma}

\begin{proof}
Below we denote by $c$ a generic constant whose value might change from line to line. By Proposition \ref{prop:magnetic-virial}, the Cauchy-Schwarz inequality and Corollary \ref{cor-uc},
\begin{align*}
\La\psi_F,  \, [H_{A,V}, i D] \, \psi_F \Ra & \geq   \|\sigma\cdot (P-A)\psi_F\|_2^2 -2  \|\sigma\cdot (P-A)\psi_F\|_2 \big(\|\wti{B} \psi_F\|_2 +\|x V^s\psi_F\|_2\big)\\ 
& \ \  -( \alpha_2+ 3\, \alpha_3) \|\sigma\cdot (P-A)\psi_F\|_2^2 - c \|\psi_F\|_2^2\, .
\end{align*}
Now let $\kappa>0$ and split 
$$
\|\sigma\cdot (P-A)\psi_F\|_2^2= (1-\kappa)\|\sigma\cdot (P-A)\psi_F\|_2^2 +\kappa \|\sigma\cdot (P-A)\psi_F\|_2^2.
$$
Using equation \eqref{eq-energy boost-2} together with Corollary \ref{cor-uc} we find 
\begin{align*}
\La\psi_F,  \, [H_{A,V}, i D] \, \psi_F \Ra & \geq (2 -\kappa) \|\sigma\cdot(P-A)\psi_F\|_2^2  + \kappa\,  \La\psi_F, |\nabla F|^2\, \psi_F \Ra- (\alpha_2 +d\alpha_3+\kappa \alpha_0) \, \|\sigma\cdot(P-A)\psi_F\|_2^2 \\
&\ \  -2  \|\sigma\cdot(P-A)\psi_F\|_2 \big(\|\wti{B} \psi_F\|_2 +\|x V^s\psi_F\|_2\big)  - c \|\psi_F\|_2^2, 
\end{align*}
and 
\begin{align*}
2 \|\sigma\cdot(P-A)\psi_F\|_2 \big(\|\wti{B} \psi_F\|_2+\|x V^s\psi_F\|_2\big)   & \, \leq \, \alpha_ 1  \|\sigma\cdot(P-A)\psi_F\|_2^2 +2\, C_1 \|(P-A)\psi_F\|_2 \,  \|\psi_F\|_2 \\
&\,  \leq \, (\alpha_ 1+\kappa)\,  \|\sigma\cdot(P-A)\psi_F\|_2^2 + \frac{C_1}{\kappa} \,  \|\psi_F\|_2^2. 
\end{align*}
Hence 
\begin{align*}
\La\psi_F,  \, [H_{A,V}, i D] \, \psi_F \Ra & \geq (1 -2\kappa-\kappa \alpha_0 -\alpha_1-\alpha_2-3 \alpha_3) \, \|\sigma\cdot(P-A)\psi_F\|_2^2 + \kappa\,  \La\psi_F, |\nabla F|^2\, \psi_F \Ra - c_\kappa \|\psi_F\|_2^2 \, .
\end{align*}
If we now set $\kappa = (2+\alpha_0)^{-1} (1-\alpha_1-\alpha_2- 3 \alpha_3)$,
then $\kappa >0,$ see \eqref{alphas}, and the claim follows. 
\end{proof}

\subsection{The  virial }\label{ssec-virial} Below we provide a matrix version of the Kato form of the virial. 

\begin{lemma}\label{lem-kato-virial}
Let $B$ and $V$ satisfy Assumptions \ref{ass-B-mild-int}-\ref{ass-V-form small}. Suppose that $V$ and  $|x|^2V^2$ are relatively form bounded with respect to $(P-A)^2$. Then 	\begin{equation}\label{eq:kato-virial}
		\La \varphi , x\cdot\nabla V\varphi \Ra = 2\im \La xV\varphi, (P-A)\varphi \Ra  - 3 \La\varphi, V\varphi \Ra
	\end{equation}
	for all $\varphi\in\D(P-A)$. 
\end{lemma}

\begin{proof}
Let $ V_{jk}$ be the matrix elements of $V$ and let $W:\R^3\to \R$. 
By \cite[Lemma~3.12]{ahk}, 
\begin{equation} \label{ahk-3.12}
\La u , x\cdot\nabla W\, u \Ra_{L^2(\R^3)} = 2\im \La  u, x\,  W\, (P-A) u \Ra_{L^2(\R^3)}   - 3 \La u, W \, u \Ra_{L^2(\R^3)} 
\end{equation}
holds for all $u\in \h^1(\R^3)$ provided $W$ and $|x|^2 W^2$ are relatively form bounded with respect to $(P-A)^2$ in $L^2(\R^3)$.  To prove the statement of the lemma 
we have to verify that equation \eqref{ahk-3.12} can be applied with $W=V_{11}, W=V_{22}$ and $W=V_{12}$, cf.~\eqref{V-hermit}. Reasoning in the same way as in the proof of Proposition 
\ref{prop:magnetic-virial} we  verify that $V_{11}, V_{22}$ and $V_{12}$ are relatively form bounded with respect to $(P-A)^2$ in $L^2(\R^3)$. In order to verify the relative form boundedness of 
$|x|^2 V^2_{11}, |x|^2 V^2_{22}$ and $|x|^2 |V_{12}|^2$ we note that since 
$$
V^2 = 
 \left( \begin{array}{lr}
V_{11}^2 +   |V_{12}|^2 & V_{12} (V_{11} +V_{22})  \\
 V_{21} (V_{11} +V_{22})  & V_{22}^2 +   |V_{12}|^2
\end{array} \right) ,
$$
the assumptions of the lemma imply that
\begin{align} \label{x2-V2-bound}
& \la u,  |x|^2 (V^2_{11} + |V_{12}|^2) \, u\ra_{L^2(\R^3)}   +\la v,  |x|^2 ( V^2_{22}+  |V_{12}|^2) \, v\ra_{L^2(\R^3)}  +\la u,  |x|^2  V_{12} (V_{11}+V_{22}) \, v\ra_{L^2(\R^3)}  \\
& \quad +\la v, |x|^2 V_{21} (V_{11}+V_{22}) \, u\ra_{L^2(\R^3)} 
 \lesssim\, \|(P-A) u\|_{L^2(\R^3)}^2 +\|(P-A) v\|_{L^2(\R^3)}^2 +\| u\|_{L^2(\R^3)}^2+\|v\|_{L^2(\R^3)}^2 , \nonumber
\end{align} 
for all $u,v \in \h^1(\R^3)$. As in the proof of Proposition 
\ref{prop:magnetic-virial} we apply \eqref{x2-V2-bound} with $v=0$ and $u=0$ respectively, and deduce that $|x|^2 (V^2_{11} + |V_{12}|^2)$ and $|x|^2 (V^2_{22} + |V_{12}|^2)$ are relatively 
form bounded with respect to $(P-A)^2$ in $L^2(\R^3)$. Hence \eqref{ahk-3.12} holds for $W=V_{jk}$ with any $j,k =1,2$. 
In view of \eqref{polarisation}, this proves equation \eqref{eq:kato-virial}.  
\end{proof}

\begin{corollary}\label{cor-mixed-virial}
 Let $B$ and $V$ satisfy Assumptions \ref{ass-B-mild-int}-\ref{ass-V-form small}.
  Assume moreover, that the potential $V$ splits as $V=V^s+V^\ell$ 
  where  $V^s$ and  $|x|^2 (V^s)^2$ are relatively form bounded with respect to 
  $(P-A)^2$ and the distribution $ x\cdot\nabla V^\ell$ extend to a quadratic form which is form bounded with respect to $(P-A)^2$. Then  
  \begin{equation}
  	\La \varphi , x\cdot\nabla V\varphi \Ra
  		= 2\im \La xV^s \varphi, (P-A)\varphi \Ra  - 3\La\varphi, V^s \varphi \Ra
  			+ \La \varphi , x\cdot\nabla V^\ell\, \varphi \Ra
  \end{equation}
 for all $\varphi\in\D(P-A)$. 
\end{corollary}
\begin{proof} 
The claim follows from Lemmas \ref{lem:xgradv} and \ref{lem-kato-virial}. 
\end{proof}

\section{\bf The main result}
\label{sec-abs}

Once we have established the virial identities \eqref{eq-psiF-2} and \eqref{eq-energy boost-2}, we can follow the strategy of \cite{fh,fhhh,ahk}. 
This is done in two steps. First we show that if eigenfunctions corresponding to energies larger than $\Lambda$ exist, then they decay faster than exponentially. Second, we prove that such eigenfunctions have to vanish identically.    

\subsection{Super-exponential decay} \noindent  Given  $x\in\R^3,\, \lambda>0$, we set 
$$
\x_\lambda \coloneqq \sqrt{\lambda+|x|^2}\ . 
$$
If $\lambda=1$, we omit the subscript and write $ \x_1\, =\,  \x$. 

 We have

\begin{proposition} \label{prop-decay}
Assume that $B$ and $V$ satisfy Assumptions \ref{ass-V-hermit}-\ref{ass-bounded infinity} and that the magnetic field $A$ corresponding to $B$ is in the Poincar\'e 
gauge. 
Furthermore, assume that $\psi$ is a weak eigenfunction of the magnetic 
Schr\"odinger operator $H_{A,V}$ corresponding to the energy $E\in\R ,$ 
and that there exist $\ol{\mu}\ge 0$ and $\lambda>0$ such that 
$x\mapsto e^{\, \ol{\mu}\, \x_\lambda }\, \psi(x) \in L^2(\R^3,\C^2)$. If $E+\ol{\mu}^2>\Lambda$ 
with  $\Lambda$ given by \eqref{edge}, then 
\begin{equation} \label{exp-decay} 
x\mapsto e^{\, \mu  \x_\lambda  }\, \psi(x) \in L^2(\R^3,\C^2) \qquad \forall \, \mu >0, \quad \, \forall \, \lambda>0. 
\end{equation}
\end{proposition}

\noindent The proof of Proposition \ref{prop-decay}  requires some preliminaries. Obviously it suffices to prove the statement for 
$\lambda=1$. First we consider the case 
$\ol{\mu}=0$ and  choose 
\begin{equation}  \label{eq-F}
F_{\mu,\eps}(x) = \frac\mu\eps \left(1-e^{-\eps\, \x}\right)\, ,
\end{equation}
for some $\mu\ge 0$ and $\veps>0$. We have   $F_{\mu,\eps}(x) \to \mu \x$ as $\eps\to 0$. 
Moreover, the identity
\begin{equation}\label{grad-F} 
	\nabla F_{\mu,\veps} = \mu \la x\ra^{-1} e^{-\veps\la x \ra} x\, 
\end{equation} 
implies
\begin{equation}\label{eq-g}
	g_{\mu,\veps}(x) = \mu \la x \ra^{-1} e^{-\veps\la x \ra}\, .
\end{equation}
Let 
$$
\mu_* = \sup \left\{ \mu\geq 0\, : \,  e^{\mu \x   } \psi \in L^2(\R^3,\C^2)\right\}\, ,
$$
be the maximal exponential decay rate of  $\psi$. To prove \eqref{exp-decay} we have to show that
 $\mu_*=\infty$. We will argue by contradiction.
 
 \begin{lemma} \label{lem-sequences}
Suppose that  $0\le \mu_*<\infty$.  Then there exist decreasing sequences $\mu_n$ and $\veps_n $ such that $\mu_n\to\mu_*$ and $\veps_n\to0$, as $n\to\infty$, and such that,
writing $F_n\coloneqq F_{\mu_n,\veps_n}$, we have
\begin{equation}  \label{div-n}
a_n : = \| e^{ F_n}\, \psi\|_2 \ \to \ \infty \quad \text{as} \quad  n\to\infty,
\end{equation} 
 \end{lemma}

\begin{proof}
For a fixed $x$ and $\mu$ we have 
\begin{equation}
\partial_\eps F_{\mu,\veps}(x) = -\frac{\mu}{\eps^2} \big(1-(1+\eps \x) \, e^{-\eps\x}\big) \, .
\end{equation}
On the other hand, a short calculation shows that the function $t\mapsto  (1+t) e^{-t}$ is strictly decreasing on $(0,\infty)$. It follows that
$$
\partial_\eps F_{\mu,\veps}(x) <0 \qquad \forall\, \eps>0, \ \ \mu>0, \ \ x\in\R^3.
$$
Thus $F_{\mu,\veps}(x)$ is strictly decreasing in $\eps$ for any $\mu>0$ and 
$$
 F_{\mu,\veps}(x) \nearrow \mu \x \qquad \text{as} \quad \eps\searrow 0.
$$
By setting  $\mu_n =\mu_* +\frac 1n$,  we then have
\begin{equation}  \label{int-div}
\lim_{\eps\searrow 0} \,  \| e^{ F_{\mu_n, \eps}}\, \psi\|_2 = \| e^{ \mu \x }\, \psi\|_2 = +\infty \qquad \forall\, n\geq 1, 
 \end{equation}
by monotone convergence. Now we construct the sequence $\eps_n$ as follows. Take $\eps_1$ such that $\| e^{ F_{\mu_1, \eps_1}}\, \psi\|_2 >  1$, and 
for each $n\geq 2$ we choose $\eps_n < \eps_{n-1}$ so that 
$$
\| e^{ F_{\mu_n, \eps_n}}\, \psi\|_2 \, > \, n,
$$
which is possible in view of \eqref{int-div}. This proves the claim.
\end{proof} 
 
Now let 
$g_n(x)\coloneqq g_{\mu_n,\veps_n}$, and define
\begin{equation} \label{fin}
\varphi_n = \frac{ e^{ F_n}\, \psi}{\|e^{ F_n}\, \psi\|_2}\, . 
\end{equation}
Since $\mu_n\to\mu_*$, and since
\begin{equation} \label{fn-upperb}
F_n(x) \, \leq \, \mu_n \x\,  , 
\end{equation}
for any compact set $\omega\subset\R^3$ it holds
\begin{equation}\label{eq-local-vanishing-1}
	\La\varphi_n, \chi_\omega\,  \varphi_n \Ra\to  0 \quad \text{as } n\to\infty \,  
\end{equation}
where $\chi$ denotes the characteristic functions. 
Hence if $W$ is bounded and $W(x)\to 0$ as $x\to\infty$, then  
\begin{equation} \label{eq-local-vanishing-2}
	\La\varphi_n, W \varphi_n \Ra\to  0 \quad \text{as } n\to\infty. 
\end{equation}

\smallskip 

\noindent We will also need the following auxiliary Lemmas. 

\begin{lemma}\label{lem-punch-1}
	Let $F_n$, $g_n$, $\psi$, and $\varphi_n$ be given as above.  If $0<\mu_*<\infty$, then 
	\begin{equation}\label{eq-punch-1}
		\lim_{n\to\infty} \la e^{F_n}\psi, \veps_n\la x \ra  e^{F_n}\psi \ra =0\, .
	\end{equation}
	Moreover, if $0\le \mu_*<\infty$, then 
	\begin{align}
		\lim_{n\to\infty}\La \nabla F_n\varphi_n,\nabla F_n  \varphi_n\Ra 
			&= \mu_*^2  \label{eq-punch-2} 
	\end{align} 
	and 
	\begin{align}
		\lim_{n\to\infty}\La \varphi_n, \big( (x\cdot\nabla)^2g_n-x\cdot \nabla |\nabla F_n|^2\big) \varphi_n\Ra 
			&=   0  \label{eq-punch-3}	
	\end{align} 
\end{lemma}

\begin{proof}
Let  $\delta>0$. Since $\|\varphi_n\|_2=1$, it follows that 
\begin{align} \label{fi-n-eps}
	\La \varphi_n, \veps_n\la x\ra \varphi_n\Ra 
		\le \delta +  \La \varphi_n, \id_{\{\veps_n\la x \ra > \delta\}}\, \veps_n\la x\ra \varphi_n\Ra . 
\end{align}
Next we note that the mapping $t\mapsto \frac{1-e^{-t}}{t}$ is decreasing on $(0,\infty)$. Hence  
\begin{equation}
	\overline{\gamma}_\delta\coloneqq \sup_{t\ge \delta}\frac{1-e^{-t}}{t} <1.
\end{equation}
This shows that for any $x$ such that $\veps_n\la x \ra\ge\delta$ we have 
\begin{align*}
	F_n = \frac{\mu_n\la x \ra}{\veps_n\la x \ra}(1-e^{-\veps_n\la x \ra}) \le \mu_n\overline{\gamma}_\delta \la x \ra  .
\end{align*}
Let $\kappa$ be such that  $\overline{\gamma}_\delta<\kappa<1$. If $0<\mu_*<\infty$ then, by the definition of $\mu_*$, 
$\psi$ decays exponentially with any rate $\mu$ satisfying $\kappa\mu_*< \mu<\mu_*$. Since $\mu_n\gamma_\delta\to \gamma_\delta\mu_* <\kappa\mu_*$ as $n\to\infty$, 
this implies 
\begin{align*}
	\limsup_{n\to\infty}  
		\La e^{F_n}\psi, \id_{\{\veps_n\la x \ra > \delta\}}\, \la x\ra  e^{F_n}\psi \Ra 
		&\le  \limsup_{n\to\infty} \La e^{\mu_n\overline{\gamma}_\delta \la x \ra}\psi, \la x\ra  e^{\mu_n\overline{\gamma}_\delta \la x \ra}\psi \Ra 
		\le  \La e^{\kappa \mu_*\la x \ra}\psi, \la x\ra  e^{\kappa \mu_*  \la x \ra}\psi \Ra
		 <\infty .
\end{align*}
Equation \eqref{eq-punch-1} thus follows from \eqref{div-n} and \eqref{fi-n-eps}.
\smallskip

\noindent  To prove the remaining claims of the Lemma we need the identity 
\begin{equation}\label{eq-|grad F|^2}
	|\nabla F_n|^2 = \mu_n^2 \big( 1-\la x \ra^{-2} \big)e^{-2\veps_n\la x \ra}\, ,
\end{equation} 
which follows by a direct calculation from equation \eqref{grad-F}. Hence
\begin{equation}\label{eq-energy-shift-1}
  \begin{split}
	\mu_n^2-\La \nabla F_n \varphi_n, \nabla F_n \varphi_n\Ra 
		&= \mu_n^2 \left( \La \varphi_n, 
							\big(1- e^{-2\veps_n \la x \ra}\big)
						\varphi_n\Ra 
						+ \La \varphi_n, 
							\la x\ra^{-2} e^{-2\veps_n \la x \ra}	
						  \varphi_n\Ra 
					\right),
  \end{split}
\end{equation}
where we have used again the fact that  $\|\varphi_n\|_2=1$. 
If $\mu_*=\lim_{n\to \infty} \mu_n=0$, then \eqref{eq-energy-shift-1} shows
\begin{align*}
	\left| \mu_n^2-\La \nabla F_n \varphi_n, \nabla F_n \varphi_n\Ra  \right| \le 2\mu_n^2 \to 0 \quad \text{as } n\to\infty\, ,
\end{align*}
and hence  \eqref{eq-punch-2} with $\mu_*=0$. If $\mu_* >0$, then 
we insert the bound  $0\le 1- e^{-2\veps_n \la x \ra} \le 2\veps_n \la x \ra$ into  \eqref{eq-energy-shift-1} and get  
\begin{align*}
	\left| \mu_n^2-\La \nabla F_n \varphi_n, \nabla F_n \varphi_n\Ra  \right| 
		&\le \mu_n^2\left(   2 \La  \varphi_n, 
							\veps_n \la x \ra
						\varphi_n\Ra 
					+  \La \varphi_n, 
							\la x\ra^{-2}
						\varphi_n\Ra 
					\right) 
			\to 0 \quad \text{as } n\to\infty .
\end{align*}
In view of  \eqref{eq-punch-1} and \eqref{eq-local-vanishing-2} this proves \eqref{eq-punch-2} in the case $\mu_*>0$.
\smallskip

It remains to prove  \eqref{eq-punch-3}.  
From the definitions of $F_n$ and $g_n$  we deduce, after a short calculation, that 
\begin{equation}\label{eq-virial-lower-order}
	\left| (x\cdot\nabla)^2g_n-x\cdot \nabla |\nabla F|^2\right|
		\lesssim 
			\mu_n(\mu_n+1)
			\left[ \la x \ra^{-2} + \la x \ra^{-1}+ \veps_n \la x \ra + \veps_n^2 \la x \ra 
			\right] e^{-\veps_n\la x \ra}
\end{equation}
This and that boundedness of mapping $t\mapsto t e^{-t}$ on $[0,+\infty)$ implies that if $\mu_*=0$, then 
\begin{align*}
	\left|
		\La \varphi_n,\big( (x\cdot\nabla)^2g_n-x\cdot \nabla |\nabla F|^2\big) \varphi_n\Ra 
	\right|
		\lesssim \mu_n(\mu_n+1)
 \to 0  \quad \text{as } n\to\infty 
\end{align*}
If $0<\mu_*<\infty$, then we use \eqref{eq-virial-lower-order} to estimate 
\begin{align*}
	\left|
		\La \varphi_n,\big( (x\cdot\nabla)^2g_n-x\cdot \nabla |\nabla F|^2\big) \varphi_n\Ra
	\right|
		\lesssim 
			\La \varphi_n,\Big(\la x \ra^{-2} +\la x \ra^{-1}\Big)  \varphi_n\Ra 
			+ \La \varphi_n, \veps_n\la x \ra  \varphi_n\Ra 
			 \to 0  \quad \text{as } n\to\infty.
\end{align*}
Here we have used again equations \eqref{eq-punch-1} and \eqref{eq-local-vanishing-2}. This completes the proof of \eqref{eq-punch-3} and hence of the Lemma. 
\end{proof}

\begin{lemma}\label{lem-punch-2}
		Let $0\le \mu_*<\infty$ and $F_n$, $g_n$, and $\varphi_n$ be given as above. If $V$  satisfies Assumptions \ref{ass-V-form small} and  
		\ref{ass-V-vanishing infinity}, then 
  	\begin{align}
		\limsup_{n\to\infty}\, \la \varphi_n, V \varphi_n\ra = : \nu 
			&\leq 0  \label{eq-punch-2-1}\\
		\liminf_{n\to\infty}\la \sigma\cdot (P-A)\varphi_n,  \sigma\cdot(P-A) \varphi_n\ra 
			&\geq E+\mu_*^2 -\nu \, . \label{eq-punch-2-2}
	\end{align}
	Moreover, if the magnetic field $B$ satisfy Assumptions \ref{ass-B-rel-bounded} and \ref{ass-bounded infinity}, then 
	\begin{align}
		\limsup_{n\to\infty} \, |\La  \sigma\cdot\wti{B}\,  \varphi_n, \sigma\cdot (P-A)\, \varphi_n \Ra| 
			&\le \beta (E+\mu_*^2-\nu)^{1/2} \label{eq-punch-2-3}\, .
	\end{align}
	Finally, if one splits $V=V^s+V^\ell$, with $V^s$ and $V^\ell$ satisfying Assumptions
	\ref{ass-bounded infinity} and  \ref{ass-bounded unique continuation}, then 
		\begin{align}
 			\limsup_{n\to\infty} \La \varphi_n, x\cdot\nabla V \varphi_n\Ra 
			&\le  2\omega_1(E+\mu_*^2-\nu)^{1/2} + \omega_2 . \label{eq-punch-2-5}
	\end{align}
\end{lemma}

\begin{proof}
First we prove that 
\begin{equation} \label{sup-finite}
 \limsup_{n\to\infty}\, | \la \varphi_n, V \varphi_n\ra| <\infty.
\end{equation} 
Indeed, by Lemma \ref{lem-2a} and equation \eqref{ass-V0-eq} 
\begin{align*}
 | \la \varphi_n, V \varphi_n\ra|  & \leq \alpha_0 (1+\eta) \| \sigma\cdot (P-A)\varphi_n\|_2^2  +C
\end{align*}
for any $\eta >0$ and some $C>0$ independent of $n$. Using \eqref{eq-energy boost} with $F=F_n$ we then further get 
\begin{align*}
 | \la \varphi_n, V \varphi_n\ra|  & \leq \alpha_0 (1+\eta) \LL \varphi_n, (E+|\nabla F_n|^2)\varphi_n \RR  + \alpha_0 (1+\eta) | \la \varphi_n, V \varphi_n\ra| +C,
\end{align*}
and \eqref{sup-finite} follow by choosing $\eta$ small enough so that $\alpha_0 (1+\eta) <1$ and letting $n\to \infty$, see \eqref{eq-punch-2}.

 To prove \eqref{eq-punch-2-1} we let $j_m:[0,\infty)\to \R_+$, $m=1,2$, be infinitely often 
 differentiable on $(0,\infty)$ with $j_1(r)=1$ for $0\le r\le 1$, $j_1(r)>0$ for $r\le 3/2$,   $j_1(r)=0$ for 
 $r\ge  7/4$,  
 and $j_2(r)=0$ for $r\le 5/4$,   $j_2(r)>0$ for $r\ge  3/2$,    $j_2(r)=1$ for $r\ge  2$. 
 Then $\inf_{r\ge 0}(j_1^2(r)+j_2^2(r))>0$ and thus 
 \begin{align*}
 	\xi_1 \coloneqq \frac{j_1}{\sqrt{j_1^2+j_2^2}}\, , 
 	\qquad 
 	\xi_2 \coloneqq \frac{j_2}{\sqrt{j_1^2+j_2^2}} 
 \end{align*}
are infinitely often differentiable with bounded derivatives and $\xi_1^2+\xi_2^2=1$. Given $R\ge 1$ we set 
\begin{align*}
	\xi_{R_-}(x)\coloneqq \xi_1(|x|/R) ,  
	\qquad 	\xi_{R_+}(x)\coloneqq \xi_2(|x|/R) .
\end{align*}
Note that $\xi_{R_+}, \xi_{R_-}\in C^\infty(\R^3)$ that the all partial derivatives of $\xi_{R_+}$ and $\xi_{R_-}$ . Moreover, 
 $\xi_{R_-}$ has compact support, and $\supp(\xi_{R_+})\subset  \U_R^{\, c}=\{x\in\R^3: |x|\ge R\}$. 
By construction,
\begin{align*}
  \La \varphi_n, V\varphi_n \Ra_+ 
  	= \La \xi_{R_-}^2\varphi_n, V\varphi_n \Ra_+ +  \La \xi_{R_+}^2\varphi_n, V\varphi_n \Ra_+ .
\end{align*}
From \cite[Lemma~4.6]{ahk} it follows that 
$$
\sup_{R\geq 1}\sup_{n\in\N} \| (P-A\, )\xi_{R_+}\, \varphi_n\|   <\infty,  \qquad \text{and} \quad \forall\, R\geq 1\ : \limsup_{n\to\infty} \| (P-A)\, \xi_{R_-}\, \varphi_n\| = 0. 
$$
Hence a combination of Lemma \ref{lem-1}, Assumption \ref{ass-B-rel-bounded} and equation \eqref{eq-local-vanishing-2} applied with $W= \xi_{R_-}$ gives  
\begin{align}
		\sup_{R\geq 1}\sup_{n\in\N} \|\sigma\cdot (P-A\, )\xi_{R_+}\, \varphi_n\|  & <\infty,  \quad \text{\rm and}\quad  \forall\, R\geq 1\ :  \limsup_{n\to\infty} \|\sigma\cdot (P-A)\, \xi_{R_-}\, \varphi_n\| = 0.      \label{ahk-eq-1}
\end{align}
Let us now treat the terms containing $V$. Since $V$ is form bounded with respect to $(P-A)^2$,  it follows from Lemma \ref{lem-2a} and equations \eqref{eq-local-vanishing-2}, \eqref{ahk-eq-1} that for a fixed  $R\ge 1$ we have 
\begin{align*}
	 \La \xi_{R_-}^2\varphi_n, V\varphi \Ra_+
	 & = \La \xi_{R_-}\varphi_n, V\xi_{R_-}\varphi_n \Ra_+
	 	\, \lesssim\,  \|\sigma\cdot (P-A)\xi_{R_-}\, \varphi_n\|_2^2+ \|\xi_{R_-}\, \varphi_n\|_2^2
	 	\to 0\, , \text{ as } n\to\infty 
\end{align*}
Moreover, since $V_+$ vanishes at infinity w.r.t.~$(P-A)^2$, 
there exist sequences $\alpha_R, \gamma_R$ with 
$\alpha_R, \gamma_R \to 0$ as $R\to\infty$ such that 
\begin{align*}
\La \xi_{R_+}^2\varphi_n, V\varphi_n \Ra_+
		= \La \xi_{R_+}\varphi_n, V\xi_{R_+}\varphi_n \Ra_+\,
		\le\,  \alpha_R\, \|\sigma\cdot (P-A)\xi_{R_+}\, \varphi_n\|_2^2 + \gamma_R\|\xi_{R_+}\, \varphi_n\|_2^2 \, .
\end{align*}
Equation \eqref{ahk-eq-1} then shows that
\begin{align*}
		\limsup_{n\to\infty} \La \xi_{R_+}^2\, \varphi_n, V\varphi_n \Ra_+ \, 
		  \lesssim\,   \alpha_R  + \gamma_R 
		  \to 0\, , \text{ as } R\to\infty\, ,
\end{align*}
which proves \eqref{eq-punch-2-1}. Next,  from  \eqref{eq-energy boost}, \eqref{eq-punch-2} and \eqref{eq-punch-2-1} we obtain  
\begin{align*}
	\liminf_{n\to\infty} \La  \sigma\cdot(P-A)\varphi_n,  \sigma\cdot(P-A)\varphi_n \Ra 
		&= \liminf_{n\to\infty} \big( E + \La \nabla F_n \varphi_n, \nabla F_n \varphi_n \Ra - \La\varphi_n, V\varphi_n \Ra \big) \geq E+\mu_*^2 -\nu .
\end{align*}
Hence \eqref{eq-punch-2-2}. To treat the term with $|\wti{B}|$ we argue in the same way as  for $V$ and conclude that for any fixed $R$, 
\begin{align*}
	\limsup_{n\to\infty}  \La\varphi_n,|\wti{B}|^2\varphi_n\Ra
	\le 	\limsup_{n\to\infty}  \La  \xi_{R_+}^2\, \varphi_n|\wti{B}|^2,\,  \varphi_n\Ra  
	\, \lesssim\, \veps_R +\beta_R^2 ,   
\end{align*}
where we used Assumption \ref{ass-bounded infinity} and equation \eqref{ahk-eq-1}. 
Since $\veps_R\to 0$ and $\beta_R\to\beta$, as $R\to\infty$, with the help of \eqref{sigma-w} we get  
\begin{align*}
	\limsup_{n\to\infty} \|\sigma \cdot\wti{B}\, \varphi_n\|_2 \le \beta \, . 
\end{align*}
Moreover, equations \eqref{eq-punch-2} and \eqref{eq-energy boost} imply
\begin{equation} 
\limsup_{n\to\infty}  \| \sigma\cdot (P-A)\, \varphi_n \|_2\, \leq \, \sqrt{E+ \mu_*^2 -\nu }\ .
\end{equation} 
Hence
$$ |\La \sigma\cdot\wti{B}\,  \varphi_n, \sigma\cdot (P-A)\, \varphi_n \Ra| 
\le  \| \sigma\cdot \wti{B}\,  \varphi_n\|_2\,  \| \sigma\cdot (P-A)\, \varphi_n \|_2 \leq \beta  \sqrt{E+ \mu_*^2 -\nu }, 
$$  
which proves \eqref{eq-punch-2-3}.   
If the potential splits as $V=V^s+V^\ell$ with $V^s,V^\ell$ satisfying Assumptions \ref{ass-bounded unique continuation} and \ref{ass-bounded infinity}, then one can argue exactly as above to conclude with   
\begin{align*}
	\limsup_{n\to\infty} |\La xV^s\varphi_n,\sigma\cdot (P-A)\varphi_n\Ra| 
		&\le  \ \omega_1 \sqrt{E+ \mu_*^2 -\nu } \ \qquad\text{\rm and} \qquad 
		\limsup_{n\to\infty} |\La \varphi_n, x \cdot\nabla V^\ell\varphi_n\Ra| 
		   \le  \ \omega_2\, .
\end{align*}
Moreover, if $V^s$ and $V^\ell$ satisfying Assumptions \ref{ass-bounded unique continuation} and 
	\ref{ass-bounded infinity},and 
$\varphi\in\D(P-A)$ with $\supp(\varphi)\subset \{|x|\ge R\}$, then using Lemma \ref{lem-2a} we get
\begin{align*}
	 |\La \varphi,  V^s\, \varphi\Ra| 
	 	&=   |\La |x|^{-1}\varphi, |x| V^s\varphi\Ra| 
	 		\le \| |x|^{-1}\varphi\|_2 \||x| V^s\varphi\|_2
	 		\lesssim R^{-1}\|\varphi\|_2 \left( \| \sigma\cdot (P-A)\varphi\|_2^2 + \|\varphi\|_2^2 \right)^{1/2}\, .
\end{align*}
Thus $\lim_{n\to\infty}\La\varphi_n, V^s\, \varphi_n\Ra=0,$ and Corollary \ref{cor-mixed-virial} gives
  \begin{equation*}
  	\limsup_{n\to\infty} 
  		\La \varphi , x\cdot\nabla V\varphi \Ra
  		\le 2\omega_1(E+\mu_*^2-\nu)^{1/2} +\omega_2\, .
  		\qedhere
  \end{equation*}
\end{proof}

\begin{proof}[\bf Proof of Proposition \ref{prop-decay}]
Assume that $0\le \mu_*^2<\infty$. One easily verifies that $F_n$ and $g_n$ satisfy the assumptions of Lemma \ref{lem-com-ef}. 
The latter in combination with equation \eqref{eq-punch-3} shows that
\begin{equation}\label{eq-contradiction-1}
	\limsup_{n\to\infty} \La\varphi_n,  \, [H, i D] \, \varphi_n \Ra\le 0\, . 
\end{equation}
On the other hand, equation \eqref{eq-hd} and Lemma \ref{lem-punch-2} imply the lower bound
\begin{equation*}
  \begin{split}
	\liminf_{n\to\infty}  \La\varphi_n,  \, [H, iD] \, \varphi_n \Ra
		&\ge 2(E+\mu_*^2-\nu) -2(\beta+\omega_1) (E+\mu_*^2-\nu)^{1/2} - \omega_2\\
		&= 2\left[ \left(\sqrt{E+\mu_*^2-\nu}-\frac{\beta+\omega_1}{2}\right)^2 -\left( \frac{\beta+\omega_1}{2}\right)^2 -\frac{\omega_2}{2}\right] .
  \end{split}
\end{equation*}
Hence if
$$
\sqrt{E+\mu_*^2-\nu}\ > \ \frac{1}{2}\big(\beta+\omega_1 + \sqrt{(\beta+\omega_1)^2 +2\omega_2}\, \big)= \sqrt{\Lambda},
$$ 
then 
\begin{equation*}
	\liminf_{n\to\infty}  \La\varphi_n, i \, [H, D] \, \varphi_n \Ra  >0 ,
\end{equation*}
which contradicts \eqref{eq-contradiction-1}. 
Thus $\mu_*=\infty$ and  \eqref{exp-decay} follows.
\end{proof}

\smallskip

\subsection{Absence of embedded eigenvalues}
We are now in position to prove our main result.
\begin{theorem} \label{thm-abs}
  Let $B$ and $V$ satisfy Assumptions \ref{ass-V-hermit}-\ref{ass-bounded unique continuation}.
  Then the Pauli operator $H_{A,V}$ has 
  no eigenvalues in the interval $(\Lambda,\infty)$, where 
  $\Lambda$ is given by \eqref{edge}.  
 \end{theorem}

\begin{proof}
Assume that $E\La\varphi,\psi\Ra = Q_{A,V}(\varphi,\psi)$ holds for all $\varphi\in\D(Q_{A,V})=\D(P-A)$, and that  $E>\Lambda$. From Proposition \ref{prop-decay} we then deduce that 
\begin{equation*}
x\mapsto e^{\, \mu  \x_\lambda  }\, \psi(x) \in L^2(\R^3,\C^2) \qquad \forall \, \mu >0, \quad \, \forall \, \lambda>0 ,
\end{equation*}
where $  \x_\lambda= (\lambda+|x|^2)^{1/2}$. Let $\mu >0, \eps>0, \lambda>0$, and define 
$$
F(x) = F_{\mu,\eps,\lambda}(x) =  \frac\mu\eps \left(1-e^{-\eps \x_\lambda}\right)\, .
$$
Then 
$$
\nabla F_{\mu,\eps,\lambda}(x) = x g_{\mu,\eps,\lambda}(x), \qquad \text{with} \qquad   g_{\mu,\eps,\lambda}(x) = \frac{\mu\, e^{-\eps \x_\lambda}}{\sqrt{\lambda+|x|^2}}\, .
$$
Let $\psi_{\mu,\eps,\lambda} = e^{F_{\mu,\eps,\lambda}} \, \psi$. Lemma \ref{lem-lowerb} and equation \eqref{eq-psiF-2} give
\begin{align} \label{lowerb-3}
 \kappa\,  \La\psi_{\mu,\eps,\lambda}, |\nabla F_{\mu,\eps,\lambda}|^2\, \psi_{\mu,\eps,\lambda} \Ra &  \leq 
 \La\psi_{\mu,\eps,\lambda}, \big((x\cdot \nabla)^2 g_{\mu,\eps,\lambda} -x\cdot \nabla |\nabla F_{\mu,\eps,\lambda}|^2  \big) \psi_{\mu,\eps,\lambda} \Ra
 + C\, \|\psi_{\mu,\eps,\lambda}\|_2^2
\end{align}
for all $\mu,\eps,\lambda >0$ and  some constant $C$ independent of $\mu, \lambda$ and $\eps$.  
Now a direct calculation shows that
\begin{equation}
\lim_{\eps\to 0} \, x\cdot \nabla |\nabla F_{\mu,\eps,\lambda}(x)|^2  = 2\lambda \mu^2 \x_\lambda^{-1} \big(1\, -\x_\lambda^{-2}\big ) \, >0\, ,
\end{equation} 
and
\begin{equation}
\lim_{\eps\to 0} \, (x\cdot \nabla)^2  g_{\mu,\eps,\lambda}(x) = -2\lambda \mu \x_\lambda^{-3} |x|^2 <0 \, .
\end{equation}
Since 
$$
\lim_{\eps\to 0} F_{\mu,\eps,\lambda}(x) := F_{\mu,\lambda}(x) = \mu \x_\lambda\, ,
$$
in view of Proposition \ref{prop-decay} we can pass to limit $\eps\to 0$ in \eqref{lowerb-3} to obtain
\begin{equation} \label{lowerb-4}
 \kappa\,  \mu^2\, \La\psi_{\mu,\lambda}, \frac{|x|^2}{\lambda+|x|^2}\, \psi_{\mu,\lambda} \Ra \, \leq 
 C\, \|\psi_{\mu,\lambda}\|_2^2 \qquad \forall \, \mu, \lambda >0 , 
\end{equation}
where 
$$
\psi_{\mu,\lambda}(x) := e^{\mu \x_\lambda}\, \psi(x) \in L^2(\R^3,\C^2)\, .
$$
Using once again Proposition \ref{prop-decay} together with the monotone convergence theorem we arrive, by letting $\lambda\to 0$, at the inequality
\begin{equation}
 \kappa\,  \mu^2\,\|\psi_{\mu}\|_2^2\, \leq 
 C\, \|\psi_{\mu}\|_2^2 \qquad \forall \, \mu >0 , 
\end{equation}
where $\psi_\mu (x)=e^{\mu |x|}\, \psi(x)$. This is of course impossible for $\mu$ large enough. 
Hence $\psi_\mu=0$ and the claim follows.
\end{proof}

\begin{remark} \label{rem-zero}
The statement of Theorem \ref{thm-abs} cannot be extend to the interval $[\Lambda, \infty)$. Indeed, the result of Loss and Yau \cite{ly} shows that 
if 
\begin{equation} \label{ly-ex}
B(x) = \frac{12}{(1+|x|^2)^3} \big (2 x_1x_3-2x_2, 2x_2x_3+2x_1, 1-x_1^2-x_2^2+x_3^2 \big), 
\end{equation} 
then zero is an eigenvalue of $H_{A,0}$. More precisely, Loss and Yau proved that there exits $A:\R^3\to 	\R^3$ with $\rt A=B$ such that 
\begin{equation} \label{zero-ef}
\sigma\cdot (P-A)\, \varphi = 0, \qquad \varphi =  \frac{1+i \sigma\cdot x}{(1+|x|^2)^{3/2}}\,  \varphi_0 ,
\end{equation}
where $\varphi_0$ is an arbitrary normalized spinor.
In this case we have  $\Lambda=0$, see equation  \eqref{ly-ex} and Lemma \ref{lem-pointwise}. Hence Theorem \ref{thm-abs} guarantees the absence of eigenvalues in the interval $(0,\infty)$. The fact that our technique cannot be applied to exclude zero eigenvalue is reflected also by the power-like decay of $\varphi$ at infinity, see \eqref{zero-ef}, which is in stark contrast to 
the super-exponential decay of eigenfunctions with positive eigenvalues, cf.~Proposition \ref{prop-decay}.

For more examples of magnetic fields supporting a zero eigenvalue we refer to \cite{amn,elt,rs}. It should be pointed out, however, that the existence of a zero eigenvalue of $H_{A,0}$ is an exceptional event. Indeed, 
it was proven in \cite{be} that  those magnetic fields for which zero is not an eigenvalue of $H_{A,0}$ form a dense set in $L^{3/2}(\R^3; \R^3)$. 
\end{remark}

As a  simple consequence of Theorem \ref{thm-abs} we obtain sufficient conditions for absence of positive eigenvalues of  $H_{A,V}$. 

\begin{corollary} \label{cor-v0}
Let $B,V$ satisfy assumptions of Proposition \ref{prop-sufficient} and suppose moreover that  $B(x) = o(|x|^{-1})$ and $\|V(x)\| = o(|x|^{-1})$ as $|x|\to \infty$. Then the operator $H_{A,V}$ has no positive eigenvalues. 
\end{corollary}

\begin{proof}
We use the splitting $V^s=V, \ V^\ell=0$. 
From the assumptions of the corollary and from Proposition \ref{prop-sufficient} we get  $\beta=\omega_1=\omega_2=0$. The claim thus follows 
from Theorem \ref{thm-abs}.
\end{proof}



\section{\bf Example}  
\label{sec-example} 

In this section we construct an example which indicates the sharpness of the critical energy $\Lambda$. Consider the radial magnetic field 
\begin{equation} \label{mag-example}
B(x) = \big(0,0, b(r) \big), \qquad b(r) = \frac{b_0}{\sqrt{1+r^2}}, \qquad r=\sqrt{x_1^2+x_2^2} \,
\end{equation} 
The vector potential associated to $B$ in the Poincar\'e gauge is then given by 
$$
A(x) = \frac{(-x_2, x_1,0)}{r} \, \int_0^r b(s) s\,  ds\,  =: \big(a_1(x_1,x_2), a_2(x_1,x_2), 0\big) .
$$
Let $v:\R \to (-\infty,0]$ be a bounded compactly supported function such that $\int_\R v <0$, and let
$$
V(x) = 
 \left( \begin{array}{cc}
0  & 0 \\
0 & v(x_3)
\end{array} \right). 
$$
Then
\begin{equation} \label{HA-eps}
H_{A,\eps V} = 
 \left( \begin{array}{cc}
h_+ \oplus P_3^2& 0 \\
0 & h_- \oplus (P_3^2 +\eps v) 
\end{array} \right) ,
\end{equation}
where $h_\pm$ are the operators in $L^2(\R^2)$ acting on their domains  as 
$$
h_\pm = (P_1-a_1)^2 + (P_2-a_2)^2  \pm b\, .
$$
Obviously, the operators $h_\pm$ are non-negative being the components of the associated two-dimensional Pauli operator. 
In addition, since $b(r)\to 0$ as $r\to\infty$, the structure of the spectra of $h_\pm$ is the same as that of the two-dimensional magnetic Schr\"odinger
operator $(P_1-a_1)^2 + (P_2-a_2)^2$. In particular, from the well-known example of Miller-Simon \cite{ms}, with a numerical error corrected in \cite[Sec.~6.1]{ahk}, it follows that 
 the spectrum of $h_\pm$ is {\em dense pure point} in $[0, b_0^2)$ and  {\em absolutely continuous} in $[b_0^2, \infty)$.

Hence if $\eps>0$ is small enough such that the operator $P_3^2 +\eps v$ in $L^2(\R)$ has exactly one discrete negative eigenvalue $-\lambda(\eps)$, then by  \eqref{HA-eps},
$$
\sigma_{\rm es} \big(H_{A,\eps V}\big) = [-\lambda(\eps), \infty),
$$
 and the spectrum of $H_{A,\eps V}$ is  {\em  dense pure point} in  $[-\lambda(\eps), b_0^2-\lambda(\eps)\big)$ and  {\em absolutely continuous} in  $\big[b_0^2-\lambda(\eps), \infty\big)$.

On the other hand, it is easily verified that $B$ and $V$ satisfy assumptions of Theorem \ref{thm-abs}, see Proposition \ref{prop-sufficient}. Moreover,  putting $V^s=0$ and $V^\ell=V$ gives  $\omega_1=\omega_2=0$. Therefore,  for any $\eps>0$ we have $\Lambda(B, \eps V)= \beta^2 =b_0^2$ . Since $\lambda(\eps) \to 0$ as $\eps\to 0$, the above example shows  that the threshold energy $\Lambda$ cannot be improved.

\begin{remark}
Examples of magnetic fields which produce  embedded eigenvalues  of $H_{A,V}$ above any fixed energy were found in \cite[Thm.~5.1]{ahs} and \cite[Thm.~3.1]{bbr}.  
\end{remark}


\section{\bf  The Dirac operator} 
\label{sec-dirac}

\noindent The  magnetic Dirac operator in $L^2(\R^3, \C^4)$ is given by 
\begin{equation} \label{dirac-op}
\mathbb{D} = \begin{pmatrix}
m \id& \sigma\cdot (P-A) \\
\sigma\cdot (P-A)  &   -m \id
\end{pmatrix}
\, ,
\end{equation}
where $m\geq 0$ is a constant. From Assumption \ref{ass-B-rel-bounded} and equation \eqref{eq-domains}
 it follows that $\D(\mathbb{D} ) = \D(P-A)$. Recall also that 
 $$
 \sigma(\mathbb{D})=  \sigma_{\rm es}(\mathbb{D}) =  (-\infty, -m]\cup[m,\infty) .
 $$
 
 We have

\begin{theorem} \label{thm-dirac}
Let $B$ satisfy the Assumptions \ref{ass-B-mild-int}, \ref{ass-B-rel-bounded}, \ref{ass-bounded infinity} and \ref{ass-bounded unique continuation} $($with $V=0)$. Suppose  that $A\in L^2_{\rm loc}(\R^3;\R^3)$ is such that ${\rm curl}\, A=B$. 
Then the Dirac operator $\mathbb{D}$ has no eigenvalues in 
$$
\big(-\infty, -\sqrt{\beta^2 +m^2} \ \big) \, \cup \big( \sqrt{\beta^2 +m^2}, \, \infty\big ),
$$
where $\beta$ is given by \eqref{def-beta-omega}.
\end{theorem}

\begin{proof}
Since the spectrum of $\mathbb{D}$ is gauge invariant, we  may suppose without loss of generality that $A$ is given by \eqref{eq-a}. Note that 
\begin{equation} 
\mathbb{D}^2 = 
 \begin{pmatrix}
H_{A,0}+ m^2\id  & 0\\
0 &   H_{A,0}+ m^2\id 
\end{pmatrix}
\, , \\[3pt]
\end{equation}
 in the sense of quadratic  forms on $\D(P-A)$. This means that if $\mathbb{D}\, \psi= E \psi$ for some $\psi\in\D(P-A)$, then $\psi$ is a weak eigenfunction of $H_{A,0}$ relative to eigenvalue $E^2-m^2$. Since $\Lambda=\beta^2$,  in view of Theorem \ref{thm-abs} we must have $E^2-m^2 \leq \beta^2 $.
\end{proof}

\begin{corollary} \label{cor-dirac} 
Let $B$ be such that $|\wti B|\in L^p_{\rm loc} (\R^3)$ for some $p>3$, and suppose that $B(x) = o(|x|^{-1})$ as $|x|\to\infty$. Then the  operator $\mathbb{D}$ has no eigenvalues in 
$(-\infty, - m ) \, \cup ( m, \, \infty)$.
\end{corollary}

\begin{proof}
This is a combination of Proposition \ref{prop-sufficient} and Theorem \ref{thm-dirac}.
\end{proof}

\begin{remark} \label{rem-zero-dirac} 
As in the case of Pauli operator we note that the claim of Corollary \ref{cor-dirac} cannot be extended to the set $(-\infty, - m ] \, \cup [ m, \, \infty)$. Indeed, the magnetic field given by  \eqref{ly-ex} satisfies assumptions of Corollary \ref{cor-dirac}, but the associated Dirac operator $\mathbb{D}$ has eigenvalues $m$ and $-m$. To see this, consider the spinor $\varphi\in L^2(\R^3,\C^2)$ given by \eqref{zero-ef}. Then, with a slight abuse of notation, 
$$
\mathbb{D} 
\begin{pmatrix}
\varphi \\
0 
\end{pmatrix}  = m \begin{pmatrix}
\varphi \\
0
\end{pmatrix}  \qquad \text{and} \qquad 
\mathbb{D} 
\begin{pmatrix}
0 \\
\varphi 
\end{pmatrix}  = -m 
\begin{pmatrix}
0 \\
\varphi
\end{pmatrix} .
$$

One should mention that sufficient conditions for the absence of all eigenvalues of $\mathbb{D}$ were established in \cite{cfk}.  Indeed, it was proved there that when $A\in W^{1,3}_{\rm loc}(\R^3)$, then the operator  $\mathbb{D}$ has no eigenvalues in 
$(-\infty, - m ] \, \cup [ m, \, \infty)$, and therefore no eigenvalues at all, if the functional inequality 
\begin{equation} \label{cfk} 
\int_{\R^3} |x|^2 |B|^2\, |u|^2 \  \leq \ c^2\!\int_{\R^3} |(P-A) \,u|^2\\[3pt]
\end{equation} 
holds for all $u\in C_0^\infty(\R^3)$ and with a constant $c$ which satisfies 
$$
c\, \Big( 11 +\frac{3^{3/2}}{2}\, \sqrt{c}\ \Big) <1,
$$
see \cite[Thm.~3.6]{cfk}.

\end{remark}

\begin{remark} 
Non existence of eigenvalues of the perturbed Dirac operator $\mathbb{D} + \id q$ was studied already by Kalf \cite{kalf}. He proved that if
\begin{equation} \label{kalf}
|x| \big( |q(x)| +|B(x)| \big) \to 0 \qquad \text{as} \ \ |x| \to \infty,
\end{equation}
then the operator $\mathbb{D} +\id  q$ has no eigenvalues in $\R\setminus [-m,m]$. Note that \eqref{kalf} implies $\beta=0$. The result of Kalf was later extended to matrix valued potentials in \cite{bg}.

\end{remark}


\appendix


\section{\bf Pointwise and local $L^p$ conditions}  
\label{sec-pointwise}

Here we formulate sufficient conditions which guarantee the validity of Assumptions \ref{ass-V-vanishing infinity}, \ref {ass-V-bounded-infinity}, \ref{ass-bounded infinity}. 

\subsection{Pointwise conditions}

\begin{lemma}  \label{lem-pointwise}
Given a magnetic field $B$ and potential $V=V^s+V^\ell$ assume that 
$|\wti{B}|, |V|_{\C^2},  |xV^s|_{\C^2}$, and $|x\cdot\nabla V^\ell|_{\C^2}$ are bounded outside of a compact set, and that
$$
\lim_{|x|\to \infty}  |V_+(x)|_{\C^2} 
=0.
$$ 
Then Assumptions \ref{ass-V-vanishing infinity}, \ref {ass-V-bounded-infinity} and \ref{ass-bounded infinity} are satisfied and  
\begin{equation} \label{beta-omega-2}
\beta\le  \, \limsup_{|x|\to\infty} |\wti B(x)|, 
\quad \omega_1 \, \le  \, \limsup_{|x|\to\infty} |x\, V^s(x)|_{\C^2}, \quad\text{and }\ \  \omega_2\, \le \, 
\limsup_{|x|\to\infty}\, | (x\cdot \nabla V^\ell(x))_+ |_{\C^2}.
\end{equation}
\end{lemma}

\begin{proof}
This is a straightforward consequence of the definitions of $\beta, \omega_1$ and $\omega_2$.
 \end{proof}

 \subsection{Local $L^p$ conditions}
The conditions of Lemma \ref{lem-pointwise} can be relaxed by considering potentials which are not necessarily bounded at infinity, but which belong to $L_{\rm loc}^p(\R^3)$ for a suitable $p$. 

\begin{proposition} \label{prop-sufficient}
Let $B,V$ satisfy conditions of Lemma \ref{lem-pointwise}, and let $V$ satisfy Assumption \ref{ass-V-hermit}. Suppose moreover that $|\wti B| \in  L_{\rm loc}^p(\R^3)$ and  $|V^{s,\ell}(\cdot)|_{\C^2} \in L_{\rm loc}^p(\R^3)$ for some $p>3$. 
Then all the hypotheses of Section \ref{sec-ass} are satisfied with $\alpha_j=0$ for $j=0,1,2,3$, and the constants $\beta, \omega_1, \omega_2$ satisfy \eqref{beta-omega-2}. 
\end{proposition}

\begin{proof}
If $|\wti B| \in  L_{\rm loc}^p(\R^3)$ with $p>3$, then it is easily seen that Assumption \ref{ass-B-mild-int} holds. In view of Lemma \ref{lem-pointwise} it thus remains to prove Assumptions \ref{ass-B-rel-bounded}, \ref{ass-V-form small} and \ref{ass-bounded unique continuation}.  Given a matrix valued function $M$ on  $\R^3$ and a test function 
\eqref{phi-u-v} with $u,v \in  \D(P-A)$, we have 
\begin{align} 
\|  M \varphi \|_2^2 \, & \leq \, \|\|M(x)\|_{\C^2}\, \varphi\|_2^2 = \int_{\R^3}  \|M(x)\|^2_{\C^2} \big(|u(x)|^2+ |v(x)|^2) \, dx, \label{M-bound-1}\\
\big | \LL \varphi, M \varphi \RR \big | \, & \leq \, \LL \varphi, \|M(x)\|_{\C^2}\, \varphi \RR = \int_{\R^3}  \|M(x)\|_{\C^2} \big(|u(x)|^2+ |v(x)|^2)\, dx .\label{M-bound-2}
\end{align}

\smallskip

\noindent  So let $ u\in \D(P-A)$, and  let $W\in L_{\rm loc}^p(\R^3), \, p>3,$ be bounded outside a compact set $K\subset\R^3$.  
The compactness of the Sobolev embedding $H^1(K) \hookrightarrow L^s(K), \, 2\leq s < 6$ and the diamagnetic inequality imply that for any $\eps>0$ there exists $C_\eps$ such that 
\begin{equation} \label{eq-comp}
\left(\int_{K} |u|^s dx \right)^{\frac 2s} \ \leq \ \eps \|\nabla |u|\|_{L^2(\R^3)}^2 + C_\eps \|u\|_{L^2(\R^3)}^2 \leq \ \eps \|(P-A) u\|_{L^2(\R^3)}^2 + C_\eps \|u\|_{L^2(\R^3)}^2  \qquad \forall\, s \in [2,6) .
\end{equation} 
Equation \eqref{eq-comp} and the H\"older inequality give
\begin{align}
\| W u\|_{L^2(\R^3)}^2\  & \leq  \|W\|_{L^\infty(K^c)}^2\,   \|u\|_{L^2(\R^3)}^2 +  \|W\|^2_{L^p(K)}\,  \left(\int_{K} |u|^{p'} dx \right)^{\frac{2}{p'}} 
 \leq \eps \| (P-A) u\|_{L^2(\R^3)}^2 + C_\eps \,  \|u\|_{L^2(\R^3)}^2,   \label{bound-0}
\end{align}
where $p' \in (2,6)$ satisfies $\frac 2p + \frac{2}{p'} =1$. By the hypotheses of the proposition we can apply the above estimate with $W$ replaced by $|B|, |\wti{B}| $ and 
$\|x V^s\|_{\C^2} $ respectively. This in combination with \eqref{M-bound-1} implies Assumption \ref{ass-B-rel-bounded}  and the upper bound \eqref{ass-B2-eq} of Assumption \ref{ass-bounded unique continuation} with $\alpha_1=0$. In the same way we get
\begin{equation*} 
\int_{K} |W| |u|^2\, dx \ \leq\  \|W\|_{L^p(K)}  \big(\,  \eps\,  \| (P-A) u\|_{L^2(\R^3)}^2 + C_\eps \,  \|u\|_{L^2(\R^3)}^2\big). 
\end{equation*} 
Inserting $W= \|V\|_{\C^2}$ in the above estimate and using \eqref{M-bound-2} we obtain Assumption \ref{ass-V-form small} with $\alpha_0=0$, and estimate \eqref{ass-xV-quantitative} with $\alpha_3=0$.

\smallskip

\noindent  To prove \eqref{ass-V1-eq} consider $W$ as above and take $R$ large enough such that $K\subset \U_R$. Integration by parts yields
\begin{align} \label{xdotW}
\int_{\U_R} x\cdot \nabla W\, |u|^2 \, dx & = R  \int_{\partial\U_R}\! W  |u|^2\, dS -2 \int_{\U_R} W \left(|u|^2 + {\rm Re\, } (\bar u\, x\cdot \nabla u)\right)\, dx .
\end{align}
Now let $\eps>0$. Since the trace embedding $H^1(\U_R) \hookrightarrow L^2(\partial\U_R)$ is compact and $W\in L^\infty(\partial \U_R)$, there exists $C_\eps$ such that 
\begin{equation}  \label{boundary}
R  \int_{\partial\U_R} W |u|^2\, dS\ \leq \ \eps \int_{\U_R} |\nabla |u||^2\, dx + C_\eps \int_{\U_R} |u|^2\, dx\ \leq\ \eps \| (P-A) u\|_{L^2(\R^3)}^2 + C_\eps \,  \|u\|_{L^2(\R^3)}^2 ,
\end{equation}
where we have used also the diamagnetic inequality. 
As for the second term in \eqref{xdotW}, we note that $x\cdot \nabla u = x\cdot (P-A) u$, see \eqref{poincare}, and hence
\begin{align*}
\Big | \int_{\U_R} W\,  {\rm Re\, } (\bar u\, x\cdot \nabla u)\, dx \Big |&  = \Big | \int_{\U_R} W\,  {\rm Im\, } (\bar u\, x\cdot (P-A) u)\, dx \Big |
\leq \  R \int_{\U_R} |W|\,  |u|\, |(P-A) u| \, dx \\ 
& \leq   R  \|Wu\|_{L^2(\U_R)} \| (P-A) u\|_{L^2(\R^3)} \leq  \frac{R^2}{\sqrt{\eps}}\, \|Wu\|_{L^2(\U_R)}^2 + \sqrt{\eps}\ \| (P-A) u\|_{L^2(\R^3)}^2 \\
& \leq \sqrt{\eps}\,  \left(R^2\,  \|W\|_{L^p(\U_R)}^2 + 1\right)\, \| (P-A) u\|_{L^2(\R^3)}^2  + C_\eps\|u\|_{L^2(\R^3)}^2 . 
\end{align*}
where we have used the estimate
\begin{equation*}
\int_{\U_R} W^2\,  |u|^2\, dx \ \leq\  \|W\|_{L^p(\U_R)}^2  \big( \,  \eps\,  \| (P-A) u\|_{L^2(\R^3)}^2 + C_\eps \,  \|u\|_{L^2(\R^3)}^2\big), 
\end{equation*} 
see \eqref{bound-0}. Putting the above estimates together and using $W\in L^\infty(K^c)$ we find that  
\begin{equation*}
| \LL u, x\cdot \nabla W u\RR_{L^2(\R^3)} | \ \leq \ \eps \| (P-A) u\|_{L^2(\R^3)}^2 + C_\eps \,  \|u\|_{L^2(\R^3)}^2\, .
\end{equation*}
The polarisation identity \eqref{polarisation} now gives
$$
|  \LL v, x\cdot \nabla W u\RR_{L^2(\R^3)} | \ \leq \ \eps \big( \| (P-A) u\|_{L^2(\R^3)}^2 +\| (P-A) v\|_{L^2(\R^3)}^2\big)+ C_\eps \,  \big(\|u\|_{L^2(\R^3)}^2 +\|v\|_{L^2(\R^3)}^2\big )\, .
$$
Applying the above estimate with $W$ replaced by the matrix elements of $ V^\ell$ yields inequality \eqref{ass-V1-eq}  with $\alpha_2=0$.
\end{proof}

\subsection{Uniformly local $L^p$ conditions} In order to include potentials with stronger singularities than those allowed by Proposition \ref{prop-sufficient}, we introduce the class 
 \begin{align}
 	L^p_{\text{loc,unif}}= \Big\{ f:  \sup_{x\in\R^3}\int_{\U_1(x)} | f(y)| ^p\, dy <\infty \big\}, \qquad p> \frac32 
 \end{align}
equipped with the norm 
 \begin{align}
 	\|f\|_{L^p_{\text{loc,unif}}} =   \sup_{x\in\R^3}\Big(\int_{\U_1(x)} |f(y)|^p\, dy\Big)^{1/p}\, .
 \end{align}
 
 \begin{definition} Let $f\in  L^p_{\text{loc,unif}}$. 
We say that $g$ is equivalent to $g$ at infinity, and write $g \sim f$ if $g\in L^\infty(\R^3)$ and if
 \begin{align}
 	\limsup_{R\to\infty} \big\|\id_{\U_R^c}  (f-g ) \big\|_{L^p_{\text{loc,unif}}}=0 .
 \end{align}
 Given $f\in L^p_{\text{loc,unif}}$, we define 
 \begin{align*}
	\gamma(f) & =  \inf  \big( \, \| g\|_\infty \, : \, g \sim f\big).
\end{align*}
\end{definition}

We then have 

\begin{proposition} \label{prop-Lp-unif}
Let $B,V$ satisfy conditions of Lemma \ref{lem-pointwise}, and let $V$ satisfy Assumption \ref{ass-V-hermit}. Suppose moreover that $|\wti B| \in  L_{\rm loc,unif}^p(\R^3)$ and  $|V^{s,\ell}(\cdot)|_{\C^2} \in L_{\rm loc,unif}^p(\R^3)$ for some $p>3/2$. 
Then all the hypotheses of Section \ref{sec-ass} are satisfied with $\alpha_j=0$ for $j=0,1,2,3$, and the constants $\beta, \omega_1, \omega_2$ satisfy 
\begin{equation} 
\beta\le  \, \limsup_{|x|\to\infty} \gamma(|\wti B(x)|) ,  
\qquad \omega_1 \, \le  \gamma( |x V^s(x)|_{\C^2}), \qquad  \omega_2\, \le \, \gamma( |x\cdot \nabla V^\ell(x)|_{\C^2})
\end{equation}
\end{proposition}

Proposition \ref{prop-Lp-unif} is a matrix valued version of the results established in \cite[Sec.~A.1]{ahk}. We therefore omit the proof and refer to \cite{ahk}.

\noindent 
\textbf{Acknowledgments:}  Hynek Kova\v{r}\'{\i}k has been partially supported by Gruppo Nazionale per Analisi Matematica, la Probabilit\`a e le loro Applicazioni (GNAMPA) of the Istituto Nazionale di Alta Matematica (INdAM). 
Dirk Hundertmark acknowledges funding by the Deutsche Forschungsgemeinschaft (DFG, German Research Foundation) -- Project-ID 258734477 -- SFB 1173.


\end{document}